\documentclass{article}%
\usepackage{amssymb}
\usepackage{amsmath}
\usepackage{amsfonts}
\usepackage{graphicx}%
\providecommand{\U}[1]{\protect\rule{.1in}{.1in}}
\newtheorem{theorem}{Theorem}

\newtheorem{corollary}[theorem]{Corollary}

\newtheorem{definition}[theorem]{Definition}
\newtheorem{example}[theorem]{Example}

\newtheorem{proposition}[theorem]{Proposition}
\newtheorem{remark}[theorem]{Remark}

\newenvironment{proof}[1][Proof]{\noindent\textbf{#1.} }{\ \rule{0.5em}{0.5em}}
\begin{document}

\title{Fredholm property and essential spectrum of $3-D$ Dirac operators with regular
and singular potentials }
\author{Vladimir Rabinovich\\Instituto Polit\'{e}cnico Nacional, ESIME Zacatenco, M\'{E}XICO}
\date{}
\maketitle

\begin{abstract}
\textit{\ }We consider the $3-D$ Dirac operator with variable regular magnetic
and electrostatic potentials, and singular potentials
\begin{equation}
\mathfrak{D}_{\boldsymbol{A},\Phi,Q_{\sin}}\boldsymbol{u}(x)=\left(
\mathfrak{D}_{\boldsymbol{A},\Phi}+Q_{\sin}\right)  \boldsymbol{u}%
(x),x\in\mathbb{R}^{3}%
\end{equation}
where
\begin{equation}
\mathfrak{D}_{\boldsymbol{A},\Phi}=%
{\displaystyle\sum\limits_{j=1}^{3}}
\alpha_{j}\left(  i\partial_{x_{j}}+A_{j}(x)\right)  +\alpha_{0}m+\Phi
(x)I_{4},
\end{equation}
$Q_{\sin}=\Gamma(s)\delta_{\Sigma}$ is the singular potential with
$\Gamma(s)=\left(  \Gamma_{ij}(s)\right)  _{i,j=1}^{4}$ being a $4\times4$
matrix and $\delta_{\Sigma}$ is the delta-function with support on a surface
$\Sigma\subset\mathbb{R}^{3}$ which divides $\mathbb{R}^{3}$ on two open
domains $\Omega_{\pm}$ with the common boundary $\Sigma,$ $\boldsymbol{u}$ is
a vector-function on $\mathbb{R}^{3}$ with values in $\mathbb{C}^{4}%
,\alpha_{j},j=0,1,2,3$ are the standard $4\times4$ Dirac matrices. We
associate with the formal Dirac operator $\mathfrak{D}_{\boldsymbol{A}%
,\Phi,Q_{\sin}}$ an unbounded operator $\mathcal{D}_{\boldsymbol{A}%
,\Phi,Q_{\sin}}$ in $L^{2}(\mathbb{R}^{3},\mathbb{C}^{4})$ generated
by $\mathfrak{D}_{\boldsymbol{A},\Phi}$ with domain in $H^{1}(\Omega
_{+},\mathbb{C}^{4})\oplus H^{1}(\Omega_{-},\mathbb{C}^{4})$
consisting of functions satisfying transmission conditions on
$\Sigma.$ We consider the self-adjointness of operator
$\mathcal{D}_{\boldsymbol{A},\Phi,Q_{\sin}}$ for unbounded
$C^{2}-$uniformly
regular surfaces $\Sigma,$ and the essential spectrum of $\mathcal{D}%
_{\boldsymbol{A},\Phi,Q_{\sin}}$ if $\Sigma$ is a $C^{2}$-surfaces with conic
exits to infinity. As application we consider the electrostatic and Lorentz
scalar $\delta_{\Sigma}-$shell interactions on unbounded surfaces $\Sigma.$



\end{abstract}

\section{Introduction}

The Schr\"{o}dinger operators with singular potentials supported on surfaces
have attracted a lot of attention: for instance they are used for a
description of quantum particles interacting with charged surfaces, in
approximations of Hamiltonians of the propagation of electrons through thin
barriers. \ The mathematical problems connected with formal Schr\"{o}dinger
operators with singular potentials
\[
\mathbb{S}=-\Delta+a\delta_{\Sigma}%
\]
where $\Sigma$ is the set of Lebesgue measure $0$ in $\mathbb{R}^{n}$ include
the realization of the formal Schr\"{o}dinger operator $\mathbb{S}$ as an
unbounded operator $\mathcal{S}$ in the Hilbert space $L^{2}(\mathbb{R}^{n})$
and the study of the spectral properties of $\mathcal{S}$. \

Over the past two decades, this topic has been intensively studied and there
is extensive literature devoted to this problem \cite{ALB}, \cite{AlbKurasov},
\cite{BEHL}, \cite{BEHL1}, \cite{BusStolz}, \cite{BEL}, \cite{BEL1},
\cite{BEL3}, \cite{BEKS94}, \cite{BusStolz}, \cite{BF}, \cite{Ra2018}).

The investigation of the $3D-$Dirac operators with singular potentials
supported on compact closed surfaces in $\mathbb{R}^{3}$ was initiated only
recently in the pioneering paper \cite{A-V}, where a new approach to extension
theory of symmetric operators was employed. This research was continued in the
papers:\cite{Ourm}, \cite{OurmieresVega}, \cite{Pan}. \ A different approach
using the abstract theory of quasi-boundary triples and their Weyl functions
was proposed in \cite{BEHL} \cite{BEHL1}. \textit{\ }

In contrast to the indicated papers, we consider singular potentials with
supports on both bounded and unbounded surfaces in $\mathbb{R}^{3}.$ Our
approach to the self-adjointness of Dirac operators is based on the study of
transmission problems with parameter associated with the Dirac operators. We
introduce the Lopatinsky-Shapiro conditions for their invertibility for large
values of the parameter, and for the a priori estimates of solutions of
associated transmission problems.

Moreover we study the Fredholm properties and the essential spectrum of
transmission problems associated with the Dirac operators with singular
potentials with supports on compact surfaces and non-compact surfaces with
conical exits to infinity. For this aim we apply the local principle
\cite{Rab1972}, \cite{Sim1967} and the limit operators method, see for
instance \cite{RRS},\cite{RRS1},\cite{Ra1},\cite{Ra2}.

We consider in the paper the Dirac operators
\begin{equation}
\mathfrak{D}_{\boldsymbol{A},\Phi,Q_{\sin}}\boldsymbol{u}(x)=\left(
\mathfrak{D}_{\boldsymbol{A},\Phi}+Q_{\sin}\right)  \boldsymbol{u}%
(x),x\in\mathbb{R}^{3}, \label{0.1}%
\end{equation}
where $\boldsymbol{u}$ is a vector-function on $\mathbb{R}^{3}$ with values in
$\mathbb{C}^{4},$
\begin{equation}
\mathfrak{D}_{\boldsymbol{A},\Phi}\boldsymbol{u}(x)=\left(
{\displaystyle\sum\limits_{j=1}^{3}}
\alpha_{j}\left(  i\partial_{x_{j}}+A_{j}(x)\right)  +\alpha_{0}%
m+\Phi(x)\right)  \boldsymbol{u}(x),x\in\mathbb{R}^{3} \label{0.2}%
\end{equation}
is the $3D-$Dirac operator with variable regular magnetic and electrostatic
potentials, $\alpha_{j},j=0,1,2,3$ are the $4\times4$ Dirac matrices
satisfying the relations
\[
\alpha_{j}\alpha_{k}+\alpha_{k}\alpha_{j}=2\delta_{jk}I\,_{4}\,(j,k=0,1,2,3),
\]
$I_{4}$ is the $4\times4$ unit matrix, $\boldsymbol{A}=(A_{1},A_{2},A_{3})$ is
the variable vector-valued potential of the magnetic field $\boldsymbol{H},$
that is $\boldsymbol{H}=\boldsymbol{\nabla}\times\boldsymbol{A},$ $\Phi$ is
the variable electrostatic potential of the electric field$\mathbb{\ }%
\boldsymbol{E},$ that is $\ \boldsymbol{E}=\boldsymbol{\nabla}\Phi,$ $m$ is
the mass of electron. We use the system of coordinates for which the Planck
constant $\mathfrak{h}=1,$ the light speed $c=1,$ and the charge of electron
$e=$ $1.$ We assume that the functions $A_{j},j=1,2,3,$ and $\Phi$ belong to
$L^{\infty}(\mathbb{R}^{3}).$ The singular potential $Q_{\sin}$ in (\ref{0.1})
is $Q_{\sin}=\Gamma\delta_{\Sigma}$ where $\Gamma=\left(  \Gamma_{ij}\right)
_{i,j=1}^{4}$ is a $4\times4$ matrix-valued function with $\Gamma_{ij}$
belonging to the class $C_{b}^{1}(\Sigma)$ of continuous bounded with the
first derivatives functions on $\Sigma$ and $\delta_{\Sigma}$ is the
delta-function with support on $C^{2}-$surface $\Sigma\subset\mathbb{R}^{3}$
which divides $\mathbb{R}^{3}$ on two open domains $\Omega_{\pm}$ with the
common boundary $\Sigma.$ We assume that either $\Sigma$ is a $C^{2}%
-$connected compact surface or $\Sigma$ is a $C^{2}-$unbounded connected
surface regular at infinity.

Let $H^{1}(\mathbb{R}^{3},\mathbb{C}^{4})$ be the Sobolev space of
$\ 4-$dimensional vector-valued functions $\boldsymbol{u}$ on $\mathbb{R}^{3}%
$. We denote by $H^{1}(\Omega_{\pm},\mathbb{C}^{4})$ the spaces of
restrictions on $\Omega_{\pm}$ functions in $H^{1}(\mathbb{R}^{3}%
,\mathbb{C}^{4})\ $and $H^{1}(\mathbb{R}^{3}\diagdown\Sigma,\mathbb{C}%
^{4})=H^{1}(\Omega_{+},\mathbb{C}^{4})\oplus H^{1}(\Omega_{-},\mathbb{C}%
^{4}).$ We associate with the formal Dirac operator $\mathfrak{D}%
_{\boldsymbol{A},\Phi,Q_{\sin}}$ an unbounded in $L^{2}(\mathbb{R}%
^{3},\mathbb{C}^{4})$ operator $\mathcal{D}_{\boldsymbol{A},\Phi,a_{+},a_{-}}$
defined by the Dirac operator $\mathfrak{D}_{\boldsymbol{A},\Phi}$ with domain%
\begin{align}
dom\mathcal{D}_{\boldsymbol{A},\Phi,a_{+},a_{-}}  &  =H_{a_{+},a_{-}}%
^{1}(\mathbb{R}^{3}\mathbb{\diagdown}\Sigma,\mathbb{C}^{4}\mathbb{)}%
\label{0.4}\\
&  \mathbb{=}\left\{  \boldsymbol{u}\in H^{1}(\mathbb{R}^{3}\mathbb{\diagdown
}\Sigma,\mathbb{C}^{4}\mathbb{)}:a_{+}\mathcal{(}s)\boldsymbol{u}_{+}%
(s)+a_{-}\mathcal{(}s)\boldsymbol{u}_{-}(s)=0,s\in\Sigma\right\}  ,\nonumber
\end{align}
where $\boldsymbol{u}_{\pm}=\gamma_{\Sigma}^{\pm}\boldsymbol{u}$ and
$\gamma_{\Sigma}^{\pm}:H^{1}(\Omega_{\pm},\mathbb{C}^{4})\rightarrow
H^{1/2}(\Sigma,\mathbb{C}^{4})$ are the trace operators, $a_{\pm}%
\mathcal{(}s)$ are $4\times4$ matrices defined as
\begin{equation}
a_{+}\mathcal{(}s)=\frac{1}{2}\Gamma(s)-i\boldsymbol{\alpha}\cdot
\boldsymbol{\nu}(s),a_{-}\mathcal{(}s)=\frac{1}{2}\Gamma
(s)+i\boldsymbol{\alpha}\cdot\boldsymbol{\nu}(s) \label{0.5}%
\end{equation}
where $\boldsymbol{\alpha}\cdot\boldsymbol{\nu}(s)=\sum_{j=1}^{3}\alpha_{j}%
\nu_{j}(s),$ and $\boldsymbol{\nu}(s)=\left(  \nu_{1}(s),\nu_{2}(s),\nu
_{3}(s)\right)  ,s\in\Sigma$ is the normal vector to $\Sigma$ directed into
$\Omega_{-}.$

We associate also with the formal Dirac operator $\mathfrak{D}_{\boldsymbol{A}%
,\Phi,Q_{s}}$ the operator $\mathbb{D}_{\boldsymbol{A},\Phi,a_{+},a_{-}}$ of
the transmission problem
\begin{equation}
\mathbb{D}_{\boldsymbol{A},\Phi,a_{+},a_{-}}\boldsymbol{u}(x)=\left\{
\begin{array}
[c]{c}%
\mathfrak{D}_{\boldsymbol{A},\Phi}\boldsymbol{u}(x),x\in\mathbb{R}%
^{3}\diagdown\Sigma\\
a_{+}\mathcal{(}s)\boldsymbol{u}_{+}(s)+a_{-}\mathcal{(}s)\boldsymbol{u}%
_{-}(s)=\boldsymbol{0},s\in\Sigma
\end{array}
\right.  \label{0.6}%
\end{equation}
acting from $H^{1}(\mathbb{R}^{3}\mathbb{\diagdown}\Sigma,\mathbb{C}%
^{4}\mathbb{)}$ into $L^{2}(\mathbb{R}^{3},\mathbb{C}^{4}).$

We study in the paper the self-adjointness in $L^{2}(\mathbb{R}^{3}%
,\mathbb{C}^{4})$ of unbounded operators $\mathcal{D}_{\boldsymbol{A}%
,\Phi,a_{+},a_{-}},$ the Fredholm properties of the operators $\mathbb{D}%
_{\boldsymbol{A},\Phi,a_{+},a_{-}}:H^{1}(\mathbb{R}^{3}\mathbb{\diagdown
}\Sigma,\mathbb{C}^{4}\mathbb{)}\rightarrow L^{2}(\mathbb{R}^{3}%
,\mathbb{C}^{4}),$ and the essential spectra of operators $\mathcal{D}%
_{\boldsymbol{A},\Phi,a_{+},a_{-}}.$

The paper is organized as follows. In Sec.2 we introduce the necessary
notations and definitions. In Sec.3 we describe the realization of formal
Dirac operators with singular potentials as unbounded operators $\mathcal{D}%
_{\boldsymbol{A},\Phi,a_{+},a_{-}}$in $L^{2}(\mathbb{R}^{3},\mathbb{C}^{4})$
and also as bounded operators $\mathfrak{D}_{\boldsymbol{A},\Phi,a_{+},a_{-}}%
$of transmission problems (\ref{0.6}), and we study the conditions of
self-adjointness of $\mathcal{D}_{\boldsymbol{A},\Phi,a_{+},a_{-}}$. Our
approach is closed to the classical approach to proof of self-adjointness of
realizations of boundary value problems\ in Hilbert space (see for instance
\cite{Agran1}, \cite{LionsMagenes}). We introduced an analogue of the
Lopatinsky-Shapiro condition on $\Sigma$ for the operator $\mathfrak{D}%
_{\boldsymbol{A},\Phi,a_{+},a_{-}},$ which yields the a priori estimate
\begin{equation}
\left\Vert u\right\Vert _{H^{1}(\mathbb{R}^{3}\mathbb{\diagdown}%
\Sigma,\mathbb{C}^{4}\mathbb{)}}\leq C\left(  \left\Vert \mathfrak{D}%
_{\boldsymbol{A},\Phi}u\right\Vert _{L^{2}(\mathbb{R}^{3},\mathbb{C}^{4}%
)}+\left\Vert u\right\Vert _{L^{2}(\mathbb{R}^{3},\mathbb{C}^{4})}\right)
\label{0.7}%
\end{equation}
for compact $C^{2}-$surfaces and noncompact $C^{2}-$surfaces with regular
boundary. \ It should be noted that the estimate (\ref{0.7}) yields the
closedness of $\mathcal{D}_{\boldsymbol{A},\Phi,a_{+},a_{-}}$ in
$L^{2}(\mathbb{R}^{n},\mathbb{C}^{4}).$ We also consider the parameter
dependent operator
\begin{equation}
\mathbb{D}_{\boldsymbol{A},\Phi,a_{+},a_{-}}(\mu)=\left\{
\begin{array}
[c]{c}%
\left(  \mathfrak{D}_{\boldsymbol{A},\Phi}-i\mu I_{4}\right)  \boldsymbol{u}%
(x),x\in\mathbb{R}^{3}\diagdown\Sigma,\mu\in\mathbb{R}\\
a_{+}\mathcal{(}s)\boldsymbol{u}_{+}(s)+a_{-}\mathcal{(}s)\boldsymbol{u}%
_{-}(s)=\boldsymbol{0},s\in\Sigma
\end{array}
\right.  , \label{0.8}%
\end{equation}
and give the uniform Lopatinsky-Shapiro conditions of the invertibility of
\[
\mathbb{D}_{\boldsymbol{A},\Phi,a_{+},a_{-}}(\mu):H^{1}(\mathbb{R}%
^{3}\mathbb{\diagdown}\Sigma,\mathbb{C}^{4}\mathbb{)}\rightarrow
L^{2}(\mathbb{R}^{3},\mathbb{C}^{4})
\]
for $\left\vert \mu\right\vert $ large enough. \ This conditions are applied
for the study of self-adjointness of $\mathcal{D}_{\boldsymbol{A},\Phi
,a_{+},a_{-}}.$

As the example we consider the operator $\mathcal{D}_{\boldsymbol{A}%
,\Phi,a_{+},a_{-}}$ associated with the singular potential describing the
electrostatic and Lorentz shall interactions. The uniform Lopatinsky-Shapiro
condition for $\mathcal{D}_{\boldsymbol{A},\Phi,a_{+},a_{-}}$ is satisfied if
\begin{equation}
\inf_{s\in\Sigma}\left\vert \eta^{2}(s)-\tau^{2}(s)-4\right\vert >0.
\label{0.10}%
\end{equation}

Note that the unbounded operator $\mathcal{D}_{\boldsymbol{A},\Phi,a_{+}%
,a_{-}}$ associated with the potential (\ref{0.9}) has been considered earlier
in the paper \cite{BEHL1} for compact closed surface and real constants
$\eta,\tau.$ In this paper the authors proved that the condition $\eta
^{2}-\tau^{2}\neq4$ is sufficient for self-adjointness of $\mathcal{D}%
_{\boldsymbol{0},0,a_{+},a_{-}}.$ They also considered the degenerated case
$\eta^{2}-\tau^{2}=4.$

In Sec. 4 we study the essential spectrum of the operator $\mathcal{D}%
_{\boldsymbol{A},\Phi,a_{+},a_{-}}$ for the cases: $\Sigma$ is either a
$C^{2}-$compact surface or a $C^{2}-$surface with conic structure at infinity.
This problem is closely connected with the Fredholm properties of the operator
$\mathbb{D}_{\boldsymbol{A},\Phi,a_{+},a_{-}}$of transmission problem acting
from $H^{1}(\mathbb{R}^{3}\mathbb{\diagdown}\Sigma,\mathbb{C}^{4}%
\mathbb{)}\rightarrow L^{2}(\mathbb{R}^{3},\mathbb{C}^{4}).$ Our approach to
the investigation of Fredholmness is based on the local principle see (see for
instance \cite{Rab1972}, \cite{Sim1967}) and the limit operators method (see
for instance \cite{Ra2017}, \cite{Ra2018}, \cite{Ra2}).

As the corollary, we obtain the descriptions of the essential spectra of
unbounded operators $\mathcal{D}_{\boldsymbol{A},\Phi,a_{+},a_{-}}\ $\ as the
union of spectra of all limit operators. It should be noted that the
conditions of the Fredholmness of $\mathbb{D}_{\boldsymbol{A},\Phi,a_{+}%
,a_{-}}$ and the location of the essential spectrum of $\mathcal{D}%
_{\boldsymbol{A},\Phi,a_{+},a_{-}}$ substantial depend on the behavior at
infinity of the regular potentials $\boldsymbol{A,}\Phi,$ the surface
$\Sigma,$ and the matrix $\Gamma.$

\section{Notations}

\begin{itemize}
\item If $X,Y$ are Banach spaces then we denote by $\mathcal{B(}X,Y)$ the
space of bounded linear operators acting from $X$ into $Y$ with the uniform
operator topology, and by $\mathcal{K}(X,Y)$ the subspace of $\mathcal{B(}%
X,Y)$ of all compact operators. In the case $X=Y$ we write shortly
$\mathcal{B}(X)$ and $\mathcal{K}(X).$

\item An operator $A\in\mathcal{B(}X,Y)$ is called a Fredholm operator if
$\ \ \emph{ker}A=\left\{  x\in X:Ax=0\right\}  $, and $\emph{coker}%
A\emph{=}Y/\operatorname{Im}A$ $\ $ are finite dimensional spaces. Let
$\mathcal{A}$ be a closed unbounded operator in a Hilbert space $\mathcal{H}$
with a dense in $\mathcal{H}$ domain $dom\mathcal{A}.$ Then $\mathcal{A}$ is
called a Fredholm operator if $\emph{ker}\mathcal{A}=\left\{  x\in
dom\mathcal{A}:\mathcal{A}x=0\right\}  $ and $\emph{coker}A=\mathcal{H}%
/\operatorname{Im}\mathcal{A}$ where $\operatorname{Im}\mathcal{A}%
\mathcal{=}\left\{  y\in\mathcal{H}:y=\mathcal{A}x,x\in\mathcal{D}%
_{\mathcal{A}}\right\}  $ are finite-dimensional spaces. Note that
$\mathcal{A}$ is a Fredholm operator as unbounded operator in $\mathcal{H}$ if
and only if $\mathcal{A}:D_{\mathcal{A}}\rightarrow\mathcal{H}$ is a Fredholm
operator as a bounded operator where $dom\mathcal{A}$ is equipped by the graph
norm
\[
\left\Vert u\right\Vert _{dom\mathcal{A}}=\left(  \left\Vert u\right\Vert
_{\mathcal{H}}^{2}+\left\Vert \mathcal{A}u\right\Vert _{\mathcal{H}}%
^{2}\right)  ^{1/2},u\in dom\mathcal{A}%
\]
(see for instance \cite{Agran1}).

\item The essential spectrum $sp_{ess}\mathcal{A}$ of an unbounded operator
$\mathcal{A}$ is a set of $\lambda\in\mathbb{C}$ such that $\mathcal{A}%
-\lambda I$ is not Fredholm operator as unbounded operator, and the discrete
spectrum $sp_{dis}\mathcal{A}$ of $\mathcal{A}$ is a set of isolated
eigenvalues of finite multiplicity. It is well known that if $\mathcal{A}$ is
a self-adjoint operator then $\ sp_{dis}\mathcal{A}\mathfrak{=}sp\mathcal{A}%
\mathfrak{\diagdown}sp_{ess}\mathcal{A}\mathfrak{.}$

\item We denote by $L^{2}(\mathbb{R}^{3},\mathbb{C}^{4})$ the Hilbert space of
$4-$dimensional vector-functions $\boldsymbol{u}(x)=(u^{1}(x),u^{2}%
(x),u^{3}(x),u^{4}(x)),x\in\mathbb{R}^{3}$ with the scalar product
\[
\left\langle u,v\right\rangle =\int_{\mathbb{R}^{3}}\boldsymbol{u}%
(x)\cdot\boldsymbol{v}(x)dx
\]
where $\boldsymbol{u\cdot v=}\sum_{j=1}^{4}u_{j}\bar{v}_{j}.$

\item We denote by $H^{s}(\mathbb{R}^{3},\mathbb{C}^{4})$ the Sobolev space on
$\mathbb{R}^{3\text{ }}$ of $4-$dimensional vector-valued functions, that is
the space of distributions $\boldsymbol{u}\in\mathcal{D}^{\prime}%
(\mathbb{R}^{3},\mathbb{C}^{4})$ such that
\[
\left\Vert \boldsymbol{u}\right\Vert _{H^{s}(\mathbb{R}^{3},\mathbb{C}^{4}%
)}=\left(  \int_{\mathbb{R}^{3}}(1+\left\vert \xi\right\vert ^{2}%
)^{s}\left\Vert \boldsymbol{\hat{u}}(\xi)\right\Vert _{\mathbb{C}^{4}}^{2}%
d\xi\right)  ^{1/2}<\infty,s\in\mathbb{R}%
\]
where $\boldsymbol{\hat{u}}$ is the Fourier transform of $\boldsymbol{u}.$ If
$\Omega$ $\ $is a domain in $\mathbb{R}^{3}$ then $H^{s}(\Omega,\mathbb{C}%
^{4})$ is the space of restrictions of $\boldsymbol{u}\in H^{s}(\mathbb{R}%
^{3},\mathbb{C}^{4})$ on $\Omega$ with the norm
\[
\left\Vert \boldsymbol{u}\right\Vert _{H^{s}(\Omega,\mathbb{C}^{4})}%
=\inf_{l\boldsymbol{u}\in H^{s}(\mathbb{R}^{3},\mathbb{C}^{4})}\left\Vert
l\boldsymbol{u}\right\Vert _{H^{s}(\mathbb{R}^{3},\mathbb{C}^{4})}%
\]
where $l\boldsymbol{u}$ is an extension of $\boldsymbol{u}$ on $\mathbb{R}%
^{3}.$

\item We denote by $C_{b}(\mathbb{R}^{3}\mathbb{)},C_{b}(\Sigma)$ the classes
of bounded continuous functions on $\mathbb{R}^{3},\Sigma,$ and by $C_{b}%
^{m}(\mathbb{R}^{3}\mathbb{)},C_{b}^{m}(\Sigma)$ the classes of functions $a$
on $\mathbb{R}^{3}$ such that $\partial^{\alpha}a\in C_{b}(\mathbb{R}%
^{3}\mathbb{)},C_{b}(\Sigma)$ for all multiindices $\alpha:\left\vert
\alpha\right\vert \leq m.$

\item Let $\Omega_{\pm}$ be domains in $\mathbb{R}^{3}$ and $\Sigma$ \ be a
$C^{2}-$ surface being their common boundary. We say that $\Sigma$ is
\textit{uniformly regular} if : (i) There exists $r>0$ such that for every
point $x_{0}\in$ $S$ there exists a ball $B_{r}(x_{0})=\left\{  x\in
\mathbb{R}^{3}:\left\vert x-x_{0}\right\vert <r\right\}  $ and the
diffeomorphism $\varphi_{x_{0}}:B_{r}(x_{0})\rightarrow B_{1}(0)$ such that
\begin{align*}
\varphi_{x_{0}}\left(  B_{r}(x_{0})\cap\Omega_{\pm}\right)   &  =B_{1}%
(0)\cap\mathbb{R}_{\pm}^{3},\mathbb{R}_{\pm}^{3}=\left\{  y=(y^{\prime}%
,y_{3})\in\mathbb{R}^{3}:y_{3}\gtrless0\right\}  ,\\
\varphi_{x_{0}}\left(  B_{r}(x_{0})\cap\Sigma\right)   &  =B_{1}%
(0)\cap\mathbb{R}_{y^{\prime}}^{2},\mathbb{R}_{y^{\prime}}^{2}=\left\{
y=(y^{\prime},y_{3})\in\mathbb{R}^{3}:y_{3}=0\right\}  ;
\end{align*}
\ \ \ \ (ii) Let $\varphi_{x_{0}}^{i},i=1,2,3$ be the coordinate functions of
the mappings $\varphi.$ Then%
\[
\sup_{x_{0}\in S}\sup_{\left\vert \alpha\right\vert \leq2,x\in B_{r}(x_{0}%
)}\left\vert \partial^{\alpha}\varphi_{x_{0}}^{i}(x)\right\vert <\infty
,i=1,2,3.
\]

Note that compact closed $C^{2}-$surfaces are uniformly regular automatically.

\item $\bigskip$Let
\[
\mathfrak{D}_{0}=\boldsymbol{\alpha}\cdot\boldsymbol{D}_{x}+\alpha_{0}m=%
{\displaystyle\sum\limits_{j=1}^{3}}
\alpha_{j}D_{x_{j}}+\alpha_{0}m,D_{x_{j}}=i\partial_{x_{j}}%
\]
be the free Dirac operator (see for instance \cite{Bjorken}, \cite{Greiner},
\cite{Thaller}, \cite{Bogolubov}) where $\boldsymbol{\alpha}=\left(
\alpha_{1},\alpha_{2},\alpha_{3}\right)  ,$ $\alpha_{j},j=0,1,2,3$ are the
$4\times4$ Dirac matrices
\begin{equation}
\alpha_{0}=\left(
\begin{array}
[c]{cc}%
I_{2} & 0\\
0 & -I_{2}%
\end{array}
\right)  ,\alpha_{j}=\left(
\begin{array}
[c]{cc}%
0 & \sigma_{j}\\
\sigma_{j} & 0
\end{array}
\right)  ,j=1,2,3, \label{1.1}%
\end{equation}%
\begin{equation}
\sigma_{1}=\left(
\begin{array}
[c]{cc}%
0 & 1\\
1 & 0
\end{array}
\right)  ,\sigma_{2}=\left(
\begin{array}
[c]{cc}%
0 & -i\\
i & 0
\end{array}
\right)  ,\sigma_{3}=\left(
\begin{array}
[c]{cc}%
1 & 0\\
0 & -1
\end{array}
\right)  \label{1.2}%
\end{equation}
are the $2\times2$ Pauli matrices satisfying the relations
\begin{equation}
\sigma_{j}\sigma_{k}+\sigma_{k}\sigma_{j}=2\delta_{jk}I_{2},j,k=1,2,3.
\label{1.3}%
\end{equation}
Relations (\ref{1.3}) yield that
\begin{equation}
\alpha_{j}\alpha_{k}+\alpha_{k}\alpha_{j}=2\delta_{jk}I_{4},j,k=0,1,2,3,
\label{1.4}%
\end{equation}
where $I_{n}$ is the $n\times n$ unit matrix. Equality (\ref{1.4}) implies
that
\[
\left(  \boldsymbol{\alpha}\cdot\boldsymbol{D}_{x}\right)  ^{2}=-\Delta I_{4}%
\]
where $\Delta$ is $3D-$Laplacian. Moreover
\[
\mathfrak{D}_{0}^{2}=\left(  -\Delta+m^{2}\right)  I_{4}.
\]
It is well-known that the unbounded operator $\mathfrak{D}_{0}$ with domain
$H^{1}(\mathbb{R}^{3},\mathbb{C}^{4})$ is self-adjoint in $L^{2}%
(\mathbb{R}^{3},\mathbb{C}^{4})$ (see for instance \cite{Thaller}), and
\[
sp\mathfrak{D}_{0}=sp_{ess}\mathfrak{D}_{0}=\left(  -\infty,-\left\vert
m\right\vert \right]
{\displaystyle\bigcup}
\left[  \left\vert m\right\vert ,+\infty\right)  .
\]

\end{itemize}

\section{Realization of Dirac operators with singular potentials as unbounded
operators in $L^{2}(\mathbb{R}^{3},\mathbb{C}^{4})$}

Let $\mathfrak{D}_{\boldsymbol{A},\Phi,Q_{\sin}}=\mathfrak{D}_{\boldsymbol{A}%
,\Phi}+Q_{\sin}$ be the Dirac operator given by formulas (\ref{0.1}%
),(\ref{0.2}). We define the product $Q_{\sin}\boldsymbol{u}$ where $Q_{\sin
}=\Gamma\delta_{\Sigma}$ and $\boldsymbol{u}$ $\in H^{1}(\mathbb{R}%
^{3}\diagdown\Sigma,\mathbb{C}^{4})$ as a distribution in $\mathcal{D}%
^{\prime}(\mathbb{R}^{3},\mathbb{C}^{4})=\mathcal{D}^{\prime}(\mathbb{R}%
^{3})\otimes\mathbb{C}^{4}$ acting on the test functions $\boldsymbol{\varphi
}\in C_{0}^{\infty}(\mathbb{R}^{3},\mathbb{C}^{4})$ as (see \cite{Kur})
\begin{equation}
\left(  Q_{\sin}\boldsymbol{u}\right)  (\boldsymbol{\varphi})=\frac{1}{2}%
\int_{\Sigma}\Gamma(s)\left(  \boldsymbol{u}_{+}(s)+\boldsymbol{u}%
_{-}(s)\right)  \cdot\boldsymbol{\varphi}(s)ds. \label{2.1}%
\end{equation}

Let \ $\boldsymbol{u=}\left(  \boldsymbol{u}_{+},\boldsymbol{u}_{-}\right)
\in H^{1}(\Omega_{+},\mathbb{C}^{4})\oplus H^{1}(\Omega_{-},\mathbb{C}^{4})$,
and $\mathfrak{D}_{\boldsymbol{A},\Phi}\boldsymbol{u}=\left(  \mathfrak{D}%
_{\boldsymbol{A},\Phi}\boldsymbol{u}_{+},\mathfrak{D}_{\boldsymbol{A},\Phi
}\boldsymbol{u}_{-}\right)  \in L^{2}(\mathbb{R}^{3},\mathbb{C}^{4}).$ Then
Integrating by parts we obtain
\begin{align}
\left\langle \mathfrak{D}_{\boldsymbol{A},\Phi,Q_{s}}u,\varphi\right\rangle
_{L^{2}(\mathbb{R}^{3},\mathbb{C}^{4})}  &  =\int_{\Omega_{+}\cup\Omega_{-}%
}\mathfrak{D}_{\boldsymbol{A},\Phi}\boldsymbol{u}(x)\cdot\boldsymbol{\varphi
}(x)dx\label{2.2}\\
&  -\int_{\Sigma}i\boldsymbol{\alpha}\cdot\boldsymbol{\nu}(s)\left(
\boldsymbol{u}_{+}(s)-\boldsymbol{u}_{-}(s)\right)  \cdot\boldsymbol{\varphi
}(s)ds\nonumber\\
&  +\frac{1}{2}\int_{\Sigma}\Gamma(s)\left(  \boldsymbol{u}_{+}%
(s)+\boldsymbol{u}_{-}(s)\right)  \cdot\boldsymbol{\varphi}%
(s)ds,\boldsymbol{\varphi}\in C_{0}^{\infty}(\mathbb{R}^{3},\mathbb{C}%
^{4}).\nonumber
\end{align}
\ Formula (\ref{2.2}) \ yields that in the sense of distributions in
$\mathcal{D}^{\prime}(\mathbb{R}^{3},\mathbb{C}^{4})$
\begin{align}
\mathfrak{D}_{\boldsymbol{A},\Phi,Q_{\sin}}\boldsymbol{u}  &  =\label{2.2'}\\
&  \mathfrak{D}_{\boldsymbol{A},\Phi}\boldsymbol{u}-\left[
i\boldsymbol{\alpha}\cdot\boldsymbol{\nu}\left(  \boldsymbol{u}_{+}%
-\boldsymbol{u}_{-}\right)  -\frac{1}{2}\Gamma\left(  \boldsymbol{u}%
_{+}+\boldsymbol{u}_{-}\right)  \right]  \delta_{\Sigma},\nonumber
\end{align}
Hence $\mathfrak{D}_{\boldsymbol{A},\Phi,Q_{\sin}}\boldsymbol{u}\in
L^{2}(\mathbb{R}^{3},\mathbb{C}^{4})$ if and only if
\begin{equation}
-i\boldsymbol{\alpha}\cdot\boldsymbol{\nu}(s)\left(  \boldsymbol{u}%
_{+}(s)-\boldsymbol{u}_{-}(s)\right)  +\frac{1}{2}\Gamma(s)\left(
\boldsymbol{u}_{+}(s)+\boldsymbol{u}_{-}(s)\right)  =0 \label{2.3}%
\end{equation}
for all $s\in\Sigma.$ Condition (\ref{2.3}) can be written of the form
\begin{equation}
a_{+}\mathcal{(}s)\boldsymbol{u}_{+}(s)+a_{-}\mathcal{(}s)\boldsymbol{u}%
_{-}(s)=\boldsymbol{0,}s\in\Sigma\label{e2.6}%
\end{equation}
where $a_{\pm}\mathcal{(}s)$ are $4\times4$ matrices:
\begin{equation}
a_{+}\mathcal{(}s)=\frac{1}{2}\Gamma(s)-i\boldsymbol{\alpha}\cdot
\boldsymbol{\nu}(s),a_{-}\mathcal{(}s)=\frac{1}{2}\Gamma
(s)+i\boldsymbol{\alpha}\cdot\boldsymbol{\nu}(s). \label{2.7}%
\end{equation}

We associate with the formal Dirac operator $\mathfrak{D}_{\boldsymbol{A}%
,\Phi,Q_{\sin}}$ the unbounded in $L^{2}(\mathbb{R}^{3},\mathbb{C}^{4})$
operator $\mathcal{D}_{\boldsymbol{A},\Phi,a_{+},a_{-}}$ given by the regular
Dirac operator $\mathfrak{D}_{\boldsymbol{A},\Phi}$ with the domain%
\begin{align}
dom\mathcal{D}_{\boldsymbol{A},\Phi,,a_{+},a_{-}}  &  =H_{a_{+},a_{-}}%
^{1}(\mathbb{R}^{3}\mathbb{\diagdown}\Sigma,\mathbb{C}^{4}\mathbb{)}%
\label{e2.8}\\
&  \mathbb{=}\left\{  u\in H^{1}(\mathbb{R}^{3}\mathbb{\diagdown}%
\Sigma,\mathbb{C}^{4}\mathbb{)}:a_{+}\mathcal{(}s)\boldsymbol{u}_{+}%
(s)+a_{-}\mathcal{(}s)\boldsymbol{u}_{-}(s)=\boldsymbol{0},s\in\Sigma\right\}
.\nonumber
\end{align}

We associate also with the formal Dirac operator $\mathfrak{D}_{\boldsymbol{A}%
,\Phi,Q_{\sin}}$ the bounded operator of the transmission problem
\begin{equation}
\mathbb{D}_{\boldsymbol{A},\Phi,a_{+},a_{-}}\boldsymbol{u}(x)=\left\{
\begin{array}
[c]{c}%
\mathfrak{D}_{\boldsymbol{A},\Phi}\boldsymbol{u}(x),x\in\mathbb{R}%
^{3}\diagdown\Sigma\\
a_{+}\mathcal{(}s)\boldsymbol{u}_{+}(s)+a_{-}\mathcal{(}s)\boldsymbol{u}%
_{-}(s)=\boldsymbol{0},s\in\Sigma
\end{array}
\right.  \label{e2-9}%
\end{equation}
acting from $H^{1}(\mathbb{R}^{3}\mathbb{\diagdown}\Sigma,\mathbb{C}%
^{4}\mathbb{)}$ into $L^{2}(\mathbb{R}^{3},\mathbb{C}^{4}).$

\subsection{ Lopatinsky-Shapiro conditions}

We consider the Lopatisky condition for the transmission operator
$\mathbb{D}_{\boldsymbol{A},\Phi,a_{+},a_{-}}$ at the points $x\in\Sigma$
which provides the local a priory estimates on $\Sigma$ for the operator
$\mathbb{D}_{\boldsymbol{A},\Phi,a_{+},a_{-}}$ This condition is an analogue
of the Lopatinsky-Shapiro condition for elliptic boundary value problems (see
for instance \cite{Agran1},\cite{LionsMagenes}).

Invariance of the Dirac operator with respect to the orthogonal
transformations of the coordinate systems in $\mathbb{R}^{3}$ allows us to
consider $\mathfrak{D}_{\boldsymbol{A},\Phi,a_{+},a_{-}}$ in the local system
of coordinates $y=(y_{1},y_{2},y_{3})$ where the axis $y_{1},y_{2}$ belong to
tangent plane to $\Sigma$ at the point $x_{0}$ and the axis $y_{3}=z$ is
directed along the normal vector $\boldsymbol{\nu}$ \ to $\Sigma$ at the point
$x_{0}\in\Sigma.$

After to passing to the local system of coordinates and taking into account
the main part of the operator $\mathfrak{D}_{\boldsymbol{A},\Phi}$ we obtain
the operator $\mathbb{D}_{a_{+}(x_{0}),a_{-}(x_{0})}^{0}$ of the transmission
problem for the half-spaces $\mathbb{R}_{\pm}^{3}=\left\{  y=(y^{\prime}%
,y_{3})\in\mathbb{R}^{3}:y_{3}\gtrless0\right\}  $
\begin{align}
&  \mathbb{D}_{a_{+}(x_{0}),a_{-}(x_{0})}^{0}\boldsymbol{u(}y)\label{3.0}\\
&  =\left\{
\begin{array}
[c]{c}%
\left(  \alpha_{1}D_{y_{1}}+\alpha_{2}D_{y_{2}}+\alpha_{3}D_{y_{3}}\right)
\boldsymbol{u}(y)=\boldsymbol{0},y\in\mathbb{R}_{+}^{3}\cup\mathbb{R}_{-}%
^{3}\\
a_{+}(x_{0})\boldsymbol{u}_{+}(y^{\prime},0)+a_{-}(x_{0})\boldsymbol{u}%
_{-}(y^{\prime},0)=0,y^{\prime}\in\mathbb{R}_{y^{\prime}}^{2}%
\end{array}
\right.  ,\nonumber
\end{align}
acting from $H^{1}(\mathbb{R}^{3}\diagdown\mathbb{R}_{y^{\prime}}%
^{2},\mathbb{C}^{4})$ into $L^{2}(\mathbb{R}^{3},\mathbb{C}^{4}).$ After the
Fourier transform with respect to $y^{\prime}\in\mathbb{R}^{2}$ we obtain the
family of \ the $1-$dimensional transmission problems depending on the
parameter $\xi^{\prime}=\left(  \xi_{1},\xi_{2}\right)  \in\mathbb{R}^{2}$
\begin{align}
&  \mathbb{\hat{D}}_{a_{+}(x_{0}),a_{-}(x_{0})}^{0}(\xi^{\prime}%
)\boldsymbol{\psi(}\xi^{\prime},z)\label{3.1}\\
&  =\left\{
\begin{array}
[c]{c}%
\left(  \alpha_{1}\xi_{1}+\alpha_{2}\xi_{2}+i\alpha_{3}\frac{d}{dz}\right)
\boldsymbol{\psi}(z)=\boldsymbol{0},z\in\mathbb{R}\diagdown\left\{  0\right\}
\\
a_{+}(x_{0})\boldsymbol{\psi}_{+}(\xi^{\prime},0)+a_{-}(x_{0})\boldsymbol{\psi
}_{-}(\xi^{\prime},0)=0,\xi^{\prime}=(\xi_{1},\xi_{2})\in\mathbb{R}^{2}%
\end{array}
\right.  ,\nonumber\\
z  &  =y_{3}%
\end{align}
where
\[
\boldsymbol{\psi}\in H^{1}(\mathbb{R\diagdown}\left\{  0\right\}
,\mathbb{C}^{4})=H^{1}(\mathbb{R}_{+},\mathbb{C}^{4})\oplus H^{1}%
(\mathbb{R}_{-},\mathbb{C}^{4}),
\]
and%
\[
a_{+}(x_{0})=\frac{1}{2}\Gamma(x_{0})-i\alpha_{3},a_{-}(x_{0})=\frac{1}%
{2}\Gamma(x_{0})+i\alpha_{3}.
\]
One can prove that the operator $\mathbb{D}_{a_{+}(x_{0}),a_{-}(x_{0})}%
^{0}:H^{1}(\mathbb{R}^{3}\diagdown\mathbb{R}^{2},\mathbb{C}^{4})\rightarrow
L^{2}(\mathbb{R}^{3},\mathbb{C}^{4})$ is invertible if and only if the
operator $\mathbb{\hat{D}}_{a_{+}(x_{0}),a_{-}(x_{0})}^{0}(\xi^{\prime}%
):H^{1}(\mathbb{R\diagdown}\left\{  0\right\}  ,\mathbb{C}^{4})\rightarrow
L^{2}(\mathbb{R}^{3},\mathbb{C}^{4})$ is invertible for every $\xi^{\prime}%
\in\mathbb{R}^{2}:\left\vert \xi^{\prime}\right\vert =1.$ \ The equality
\begin{equation}
\left(  \alpha_{1}\xi_{1}+\alpha_{2}\xi_{2}+i\alpha_{3}\frac{d}{dz}\right)
^{2}=\left(  \left\vert \xi^{\prime}\right\vert ^{2}-\frac{d^{2}}{dz^{2}%
}\right)  I_{4} \label{3.1'}%
\end{equation}
yields that $\mathbb{\hat{D}}_{a_{+}(x_{0}),a_{-}(x_{0})}^{0}(\xi^{\prime})$
is the Fredholm\ operator for every $\xi^{\prime}\in S^{1}$ with index $0.$
Hence $\mathbb{\hat{D}}_{a_{+}(x_{0}),a_{-}(x_{0})}^{0}(\xi^{\prime})$ is
invertible if and only if $\ker\mathbb{\hat{D}}_{a_{+}(x_{0}),a_{-}(x_{0}%
)}^{0}(\xi^{\prime})=\left\{  0\right\}  .$ We consider the equation
\begin{equation}
\left(  \alpha_{1}\xi_{1}+\alpha_{2}\xi_{2}+i\alpha_{3}\frac{d}{dz}\right)
\boldsymbol{\psi}(\xi^{\prime},z)=0,z\in\mathbb{R},\text{ }\xi^{\prime}%
=(\xi_{1},\xi_{2}). \label{3.2}%
\end{equation}
where
\[
\boldsymbol{\psi}(\xi^{\prime},z)=\left(
\begin{array}
[c]{c}%
\boldsymbol{\psi}^{1}(\xi^{\prime},z)\\
\boldsymbol{\psi}^{2}(\xi^{\prime},z)
\end{array}
\right)
\]
be the $4D-$vector with $\boldsymbol{\psi}^{j}(\xi^{\prime},z)\in
\mathbb{C}^{2},j=1,2.$ Then (\ref{3.2}) yields that $\boldsymbol{\psi}^{j}%
(\xi^{\prime},z)$ satisfies the equation
\begin{equation}
\left(  \sigma_{1}\xi_{1}+\sigma_{2}\xi_{2}+i\sigma_{3}\frac{d}{dz}\right)
\boldsymbol{\psi}^{j}(\xi^{\prime},z)=\boldsymbol{0},j=1,2. \label{3.3}%
\end{equation}
By formula (\ref{3.1'}) equation (\ref{3.3}) has the exponential solutions
$\boldsymbol{\psi}_{\pm}^{j}(\xi^{\prime},z)=\boldsymbol{h}_{\pm}^{j}%
(\xi^{\prime})e^{\pm\left\vert \xi^{\prime}\right\vert z}$ where
$\mathbf{h}_{\pm}^{j}(\xi^{\prime})\in\mathbb{C}^{2}.$ Moreover, formula
(\ref{3.1'}) implies that $\boldsymbol{h}_{\pm}^{j}(\xi^{\prime})=\Lambda
_{\pm}(\xi^{\prime})\boldsymbol{f}_{\pm}$ $\ $where
\begin{align*}
\Lambda_{\pm}(\xi^{\prime})  &  =\sigma_{1}\xi_{1}+\sigma_{2}\xi_{2}\pm
i\sigma_{3}\left\vert \xi^{\prime}\right\vert =\left(
\begin{array}
[c]{cc}%
\pm i\left\vert \varsigma\right\vert  & \bar{\varsigma}\\
\varsigma & \mp i\left\vert \varsigma\right\vert
\end{array}
\right) \\
\varsigma &  =\xi_{1}+i\xi_{2}\in\mathbb{C}^{2}.
\end{align*}
and $\boldsymbol{f}_{\pm}\in\mathbb{C}^{2}$ are arbitrary vectors. Note that
for every vectors $\boldsymbol{f}_{1}\boldsymbol{,f}_{2}\in\mathbb{C}^{2}$%
\begin{equation}
\Lambda_{\pm}(\xi^{\prime})\boldsymbol{f}_{1}\cdot\Lambda_{\mp}\left(
\xi^{\prime}\right)  \boldsymbol{f}_{2}=0. \label{3.4}%
\end{equation}

Let
\begin{align}
\boldsymbol{\psi}_{+,1}(\xi^{\prime},z)  &  =\boldsymbol{h}_{+,1}(\xi^{\prime
})e^{\left\vert \xi^{\prime}\right\vert z},\boldsymbol{\psi}_{+,2}(\xi
^{\prime},z)=\boldsymbol{h}_{+,2}(\xi^{\prime})e^{\left\vert \xi^{\prime
}\right\vert z},\label{3.5}\\
\boldsymbol{\psi}_{-,1}(\xi^{\prime},z)  &  =\boldsymbol{h}_{-,1}(\xi^{\prime
})e^{-\left\vert \xi^{\prime}\right\vert z},\boldsymbol{\psi}_{-,2}%
(\xi^{\prime},z)=\boldsymbol{h}_{-,2}(\xi^{\prime})e^{-\left\vert \xi^{\prime
}\right\vert z},\nonumber
\end{align}
where
\begin{equation}
\boldsymbol{h}_{\pm,1}(\xi^{\prime})=\left(
\begin{array}
[c]{c}%
\Lambda_{\pm}(\xi^{\prime})\boldsymbol{e}\\
\boldsymbol{0}%
\end{array}
\right)  ,\boldsymbol{h}_{\pm,2}(\xi^{\prime})=\left(
\begin{array}
[c]{c}%
\boldsymbol{0}\\
\Lambda_{\pm}(\xi^{\prime})\boldsymbol{e}%
\end{array}
\right)  ,\boldsymbol{e}=\left(
\begin{array}
[c]{c}%
1\\
0
\end{array}
\right)  ,\boldsymbol{0=}\left(
\begin{array}
[c]{c}%
0\\
0
\end{array}
\right)  . \label{3.7}%
\end{equation}
Formula (\ref{3.4}) yields that the system of vectors
\[
\left\{  \boldsymbol{h}_{+,1}(\xi^{\prime}),\boldsymbol{h}_{+,2}(\xi^{\prime
}),\boldsymbol{h}_{-,1}(\xi^{\prime}),\boldsymbol{h}_{-,2}(\xi^{\prime
})\right\}
\]
is orthogonal in $\mathbb{C}^{4}$ and
\[
\left\Vert \boldsymbol{h}_{\pm,j}(\xi^{\prime})\right\Vert _{\mathbb{C}^{2}%
}^{2}=2\left\vert \xi^{\prime}\right\vert ^{2}.
\]
Moreover, $\left\{  \boldsymbol{\psi}_{\pm,1}(\xi^{\prime},z),\boldsymbol{\psi
}_{\pm,2}(\xi^{\prime},z)\right\}  $ is the fundamental system of solutions of
equation (\ref{3.2}). Note that the general solution of equation (\ref{3.2})
in $H^{1}(\mathbb{R\diagdown}\left\{  0\right\}  ,\mathbb{C}^{4})$ is of the
form
\[
\boldsymbol{\psi}(\xi^{\prime},z)=\left\{
\begin{array}
[c]{c}%
\boldsymbol{\psi}_{+}(\xi^{\prime},z)=C_{+,1}\boldsymbol{\psi}_{-,1}%
(\xi^{\prime},z)+C_{+,2}\boldsymbol{\psi}_{-,2}(\xi^{\prime},z),z\geq0\\
\boldsymbol{\psi}_{-}(\xi^{\prime},z)=C_{-,1}\boldsymbol{\psi}_{+,1}%
(\xi^{\prime},z)+C_{-,2}\boldsymbol{\psi}_{+,2}(\xi^{\prime},z),z<0
\end{array}
\right.  .
\]
Substituting $\boldsymbol{\psi}_{\pm}(\xi^{\prime},z)$ in the interaction condition%

\begin{equation}
a_{+}(x_{0})\boldsymbol{\psi}_{+}(\xi^{\prime},0)+a_{-}(x_{0})\boldsymbol{\psi
}_{-}(\xi^{\prime},0)=\boldsymbol{0} \label{3.7'}%
\end{equation}
we obtain the linear system of equations with respect to $C_{+,1}%
,C_{+,2},C_{-,1},C_{-,2}$
\begin{equation}
a_{+}(x_{0})\boldsymbol{h}_{-,1}C_{+,1}+a_{+}(x_{0})\boldsymbol{h}%
_{-,2}C_{+,2}+a_{-}(x_{0})\boldsymbol{h}_{+,1}C_{-,1}+a_{-}(x_{0}%
)\boldsymbol{h}_{+,1}C_{-,2}=\boldsymbol{0}. \label{3.8}%
\end{equation}
Note that $\ker\mathbb{\hat{D}}_{a_{+}(x_{0}),a_{-}(x_{0})}^{0}(\xi^{\prime
})=\left\{  0\right\}  $ if and only if system (\ref{3.8}) has the trivial
solution only.

We denote by $\mathcal{L(}x_{0},\xi^{\prime})$ the $4\times4$ matrix with
columns%
\[
\left(  a_{+}(x_{0})\boldsymbol{h}_{-,1}(\xi^{\prime}),a_{+}(x_{0}%
)\boldsymbol{h}_{-,2}(\xi^{\prime}),a_{-}(x_{0})\boldsymbol{h}_{+,1}%
(\xi^{\prime}),a_{-}(x_{0})\boldsymbol{h}_{+,2}(\xi^{\prime})\right)  .
\]

\begin{definition}
We say that the local Lopatinsky-Shapiro condition for $\mathbb{D}%
_{\boldsymbol{A},\Phi,a_{+},a_{-}}$ is satisfied at the point $s_{0}\in\Sigma$
if
\begin{equation}
\det\mathcal{L(}s_{0},\xi^{\prime})\neq0\text{ for every }\xi^{\prime
}:\left\vert \xi^{\prime}\right\vert =1 \label{2.18}%
\end{equation}
and we say that the uniform Lopatinsky-Shapiro condition for $\mathbb{D}%
_{\boldsymbol{A},\Phi,a_{+},a_{-}}$is satisfied on $\Sigma$ if
\begin{equation}
\inf_{s\in\Sigma}\inf_{\left\vert \xi^{\prime}\right\vert =1}\left\vert
\det\mathcal{L(}s,\xi^{\prime})\right\vert >0. \label{2.19}%
\end{equation}

\end{definition}

Thus the operator $\mathbb{D}_{a_{+}(x_{0}),a_{-}(x_{0})}^{0}$ is invertible
if condition (\ref{2.18}) is satisfied. Moroever, if condition (\ref{2.19}) is
satisfied then the the family of the operators $\mathbb{D}_{a_{+}(s),a_{-}%
(s)}^{0},s\in\Sigma$ is uniformly invertible, that is
\begin{equation}
\sup_{s\in\Sigma,\left\vert \xi^{\prime}\right\vert =1}\left\Vert \left(
\mathbb{D}_{a_{+}(s),a_{-}(s)}^{0}\right)  ^{-1}\right\Vert <\infty.
\label{2.20}%
\end{equation}

\subsection{Examples of the Lopatinsky-Shapiro conditions}

\begin{itemize}
\item Let
\end{itemize}

\[
\Gamma=\left(
\begin{array}
[c]{cc}%
2\gamma I_{2} & 0\\
0 & 2\epsilon I_{2}%
\end{array}
\right)  ,
\]
where $\gamma=\gamma(s),\epsilon=\epsilon(s)\in C_{b}^{1}(\Sigma)$ are
real-valued functions. Then
\[
a_{\pm}=\left(
\begin{array}
[c]{cc}%
\gamma I_{2} & \mp i\sigma_{3}\\
\mp i\sigma_{3} & \epsilon I_{2}%
\end{array}
\right)  .
\]

\begin{itemize}
\item We set%
\[
\mathfrak{e}_{1}=a_{+}h_{-,1},\mathfrak{e}_{2}=ah_{-,2},\mathfrak{e}_{3}%
=a_{-}h_{+,1},\mathfrak{e}_{4}=a_{-}h_{+,2},
\]
where%
\begin{align}
\mathfrak{e}_{1}  &  =\left(
\begin{array}
[c]{cc}%
\gamma I_{2} & -i\sigma_{3}\\
-i\sigma_{3} & \epsilon I_{2}%
\end{array}
\right)  \left(
\begin{array}
[c]{c}%
\Lambda_{-}\boldsymbol{e}\\
0
\end{array}
\right)  =\left(
\begin{array}
[c]{c}%
\gamma\Lambda_{-}\boldsymbol{e}\\
i\Lambda_{+}\boldsymbol{e}%
\end{array}
\right) \label{ex1}\\
\mathfrak{e}_{2}  &  =\left(
\begin{array}
[c]{cc}%
\gamma I_{2} & -i\sigma_{3}\\
-i\sigma_{3} & \epsilon I_{2}%
\end{array}
\right)  \left(
\begin{array}
[c]{c}%
0\\
\Lambda_{-}\boldsymbol{e}%
\end{array}
\right)  =\left(
\begin{array}
[c]{c}%
i\Lambda_{+}\boldsymbol{e}\\
\epsilon\Lambda_{-}\boldsymbol{e}%
\end{array}
\right) \nonumber\\
\mathfrak{e}_{3}  &  =\left(
\begin{array}
[c]{cc}%
\gamma I_{2} & i\sigma_{3}\\
i\sigma_{3} & \epsilon I_{2}%
\end{array}
\right)  \left(
\begin{array}
[c]{c}%
\Lambda_{+}\boldsymbol{e}\\
0
\end{array}
\right)  =\left(
\begin{array}
[c]{c}%
\gamma\Lambda_{+}\boldsymbol{e}\\
-i\Lambda_{-}\boldsymbol{e}%
\end{array}
\right) \nonumber\\
\mathfrak{e}_{4}  &  =\left(
\begin{array}
[c]{cc}%
\gamma I_{2} & i\sigma_{3}\\
i\sigma_{3} & \epsilon I_{2}%
\end{array}
\right)  \left(
\begin{array}
[c]{c}%
0\\
\Lambda_{+}\boldsymbol{e}%
\end{array}
\right)  =\left(
\begin{array}
[c]{c}%
-i\Lambda_{-}\boldsymbol{e}\\
\epsilon\Lambda_{+}\boldsymbol{e}%
\end{array}
\right)  .\nonumber
\end{align}
Taking into account that
\[
\Lambda_{\pm}^{2}=0\text{ and }\Lambda_{\pm}^{\ast}=\Lambda_{\mp}%
\]
we obtain that
\begin{align}
\mathfrak{e}_{1}\cdot\mathfrak{e}_{2}  &  =\mathfrak{e}_{1}\cdot
\mathfrak{e}_{3}=\mathfrak{e}_{2}\cdot\mathfrak{e}_{4}=\mathfrak{e}_{3}%
\cdot\mathfrak{e}_{4}=0,\label{ex2}\\
\mathfrak{e}_{1}^{2}  &  =\mathfrak{e}_{2}^{2}=2(1+\gamma^{2})\left\vert
\xi^{\prime}\right\vert ^{2},\mathfrak{e}_{3}^{2}=\mathfrak{e}_{4}%
^{2}=2(1+\epsilon^{2})\left\vert \xi^{\prime}\right\vert ^{2},\nonumber\\
\mathfrak{e}_{1}\cdot\mathfrak{e}_{4}  &  =\mathfrak{e}_{2}\cdot
\mathfrak{e}_{3}=2i(\gamma+\epsilon)^{2}\left\vert \xi^{\prime}\right\vert
^{2}.\nonumber
\end{align}
Formulas (\ref{ex2}) yields that
\begin{align}
\det\left(  \mathfrak{e}_{i}\cdot\mathfrak{e}_{j}\right)  _{i,j=1}^{4}  &
=16\left\vert \xi^{\prime}\right\vert ^{8}((1+\gamma^{2})(1+\epsilon
^{2})-(\gamma+\epsilon)^{2})\label{ex3}\\
&  =16\left\vert \xi^{\prime}\right\vert ^{8}\left(  1-\gamma\epsilon\right)
^{2}.\nonumber
\end{align}
Hence the Lopatinsky-Shapiro condition holds at the point $s\in\Sigma$ \ if
\begin{equation}
\gamma(s)\epsilon(s)\neq1, \label{ex4'}%
\end{equation}
and the Lopatinsky-Shapiro condition holds uniformly on $\Sigma$ if
\begin{equation}
\inf_{s\in\Sigma}\left\vert 1-\gamma(s)\epsilon(s)\right\vert >0. \label{ex4}%
\end{equation}

\item For the electrostatic potential \ $\Gamma(s)=\tau(s)I_{4}$
$\gamma(s)=\epsilon(s)=\frac{1}{2}\tau(s),$ we obtain from (\ref{ex4'}%
),(\ref{ex4}) the local Lopatisky condition
\[
\tau^{2}(s)\neq4
\]
and the uniform Lopatinsky-Shapiro condition
\[
\inf_{s\in\Sigma}\left\vert 4-\tau^{2}(s)\right\vert >0.
\]

\item Let
\begin{equation}
Q_{\sin}=\left(  \eta(s)I_{4}+\tau(s)\alpha_{0}\right)  \delta_{\Sigma
}=\left(
\begin{array}
[c]{cc}%
\left(  \eta(s)+\tau(s)\right)  I_{2} & 0\\
0 & \left(  \eta(s)-\tau(s)\right)  I_{2}%
\end{array}
\right)  \delta_{\Sigma} \label{ex5}%
\end{equation}
be the sum of electrostatic and Lorentz potentials where $\eta(s),\tau(s)\in
C_{b}(\Sigma)$ are real-valued functions. Then we obtain from formulas
(\ref{ex4'}),(\ref{ex4}) the local Lopatisky-Shapiro condition
\[
\eta^{2}(s)-\tau^{2}(s)\neq4,s\in\Sigma
\]
and the uniform Lopatinsky-Shapiro condition
\begin{equation}
\inf_{s\in\Sigma}\left\vert \eta^{2}(s)-\tau^{2}(s)-4\right\vert >0.
\label{ex6}%
\end{equation}

\end{itemize}

\subsection{A priori estimates for operators $\mathbb{D}_{\boldsymbol{A}%
,\Phi,a_{+},a_{-}}$}

\begin{theorem}
\label{t2.1} Let $\Sigma\subset\mathbb{R}^{3}$ be a $C^{2}$-uniformly regular
surface being the common boundary of the domains $\Omega_{\pm},$
$\boldsymbol{A=}(A_{1},A_{2},A_{3})\in C_{b}^{1}(\mathbb{R}^{3},\mathbb{C}%
^{3}),$ $\Phi\in C_{b}^{1}(\mathbb{R}),$ $\Gamma\in C_{b}^{1}(\Sigma
,\mathcal{B(}\mathbb{C}^{4}))=C_{b}^{1}(\Sigma)\otimes\mathcal{B(}%
\mathbb{C}^{4}),$ and the Lopatinsky-Shapiro condition be satisfied on
$\Sigma$ uniformly, that is%
\begin{equation}
\inf_{s\in\Sigma,\left\vert \xi^{\prime}\right\vert =1}\left\vert
\det\mathcal{L(}s,\xi^{\prime})\right\vert >0. \label{2.19''}%
\end{equation}
Then there exists $C>0$ such that for every $u\in H^{1}\left(  \mathbb{R}%
^{3}\mathbb{\diagdown}\Sigma,\mathbb{C}^{4}\right)  $
\begin{equation}
\left\Vert u\right\Vert _{H^{1}\left(  \mathbb{R}^{3}\mathbb{\diagdown}%
\Sigma,\mathbb{C}^{4}\right)  }\leq C\left(  \left\Vert \mathbb{D}%
_{\boldsymbol{A},\Phi,a_{+},a_{-}}u\right\Vert _{L^{2}(\mathbb{R}%
^{3},\mathbb{C}^{4})}+\left\Vert u\right\Vert _{L^{2}(\mathbb{R}%
^{3},\mathbb{C}^{4})}\right)  . \label{2.20'}%
\end{equation}

\end{theorem}

\begin{proof}
The proof of the a priori estimate (\ref{2.20'}) is based on the local a
priori estimates following from the uniform ellipticity on $\mathbb{R}^{3}$
of\ Dirac operator $\mathfrak{D}_{\boldsymbol{A},\Phi}$ and the uniform
Lopatinsky-Shapiro conditions. For the gluing of local estimates we use a
countable partition of the unity of finite multiplicity. The uniform
ellipticity of $\mathfrak{D}_{\boldsymbol{A},\Phi}$ yields that there exists a
small enough $r>0$ such that for every point $x_{0}\in\mathbb{R}^{3}%
\diagdown\Sigma$ there exists a ball $B_{r}(x_{0})=\left\{  x\in\mathbb{R}%
^{3}:\left\vert x-x_{0}\right\vert <r\right\}  $ such that the local a priory
estimate
\begin{equation}
\left\Vert u\right\Vert _{H^{1}\left(  B_{r}(x_{0}),\mathbb{C}^{4}\right)
}\leq C\left(  \left\Vert \mathfrak{D}_{\boldsymbol{A},\Phi,a_{+},a_{-}%
}u\right\Vert _{L^{2}(B_{r}(x_{0}),\mathbb{C}^{4})}+\left\Vert u\right\Vert
_{L^{2}(B_{r}(x_{0}),\mathbb{C}^{4})}\right)  \label{2.21'}%
\end{equation}
holds for every $u\in H^{1}\left(  B_{r}(x_{0}),\mathbb{C}^{4}\right)  $ with
a constant $C>0$ independent of $x_{0}\in\mathbb{R}^{3}\diagdown\Sigma.$
\ \ The uniform Lopatinsky-Shapiro condition yields that for small enough
$r>0$ independent of $x_{0}\in\Sigma$ there exists the local a priori
estimates at every point $x_{0}\in\Sigma$%
\begin{equation}
\left\Vert u\right\Vert _{H^{1}\left(  B_{r}(x_{0})\diagdown\Sigma
,\mathbb{C}^{4}\right)  }^{2}\leq C\left(  \left\Vert \mathbb{D}%
_{\boldsymbol{A},\Phi,a_{+},a_{-}}u\right\Vert _{L^{2}(B_{r}(x_{0}%
),\mathbb{C}^{4})}^{2}+\left\Vert u\right\Vert _{L^{2}(B_{r}(x_{0}%
),\mathbb{C}^{4})}^{2}\right)  \label{2.21''}%
\end{equation}
for $u\in H^{1}(B_{r}(x_{0})\diagdown\Sigma,\mathbb{C}^{4})$ with a constant
$C>0$ independent of $x_{0}\in\Sigma.$ \ Since $\Sigma$ is a uniformly regular
surface there exists a partition of unity
\begin{equation}%
{\displaystyle\sum\limits_{j=1}^{\infty}}
\theta_{j}(x)=1,x\in\mathbb{R}^{3},\theta_{j}\in C_{0}^{\infty}(B_{r}(x_{j}))
\label{2.21'''}%
\end{equation}
subordinated to the countable covering $\cup_{j=1}^{\infty}B_{r}(x_{j})$ of
finite multiplicity such that a priori estimates (\ref{2.21'}) or
(\ref{2.21''}) hold for all points $x_{0}=x_{j}$ with a constant $C>0$
independent of $j\in\mathbb{N}.$ \ We obtain a priori estimate (\ref{2.20'})
gluing these estimates by means of partition of unity (\ref{2.21'''}).
\end{proof}

\subsection{ Parameter-dependent transmission problems associated with Dirac
operators with singular potentials}

We consider the invertibility of the parameter-dependent operator
\begin{align}
&  \mathbb{D}_{\boldsymbol{A},\Phi,a_{+},a_{-}}(i\mu)\boldsymbol{u}%
(x)\label{2.30}\\
&  =\left(  \mathbb{D}_{\boldsymbol{A},\Phi,a_{+},a_{-}}-i\mu I_{2}\right)
\boldsymbol{u}(x)\nonumber\\
&  =\left\{
\begin{array}
[c]{c}%
(\mathfrak{D}_{\boldsymbol{A},\Phi}-i\mu I_{4})\boldsymbol{u}(x),x\in
\mathbb{R}^{2}\diagdown\Gamma,\mu\in\mathbb{R},\\
a_{+}\mathcal{(}s)\boldsymbol{u}_{+}(s)+a_{-}\mathcal{(}s)\boldsymbol{u}%
_{-}(s)=0,s\in\Gamma
\end{array}
\right.  ,\mu\in\mathbb{R}\nonumber
\end{align}
acting from $H^{1}(\mathbb{R}^{3}\diagdown\Gamma,\mathbb{C}^{4})$ in
$L^{2}(\mathbb{R}^{3},\mathbb{C}^{2})$ for large value of $\left\vert
\mu\right\vert .$ We will apply the local approach following to the well-known
paper \cite{AgranVishik}). Let $\mathfrak{D}_{\boldsymbol{A},\Phi}^{0}%
(\mu)=\boldsymbol{\sigma\cdot D}-i\mu I_{4}$ is the main part of the
parameter-dependent operator $\mathfrak{D}_{\boldsymbol{A},\Phi}-i\mu I_{2}.$
Since
\[
(\boldsymbol{\sigma}\cdot i\boldsymbol{\nabla}+i\mu I_{2})\left(
\boldsymbol{\sigma}\cdot i\boldsymbol{\nabla}-i\mu I_{2}\right)  =\left(
-\Delta_{2}+\mu^{2}\right)  I_{2}%
\]
the operator $\mathfrak{D}_{\boldsymbol{A},\Phi}-i\mu I_{2}$ is the uniformly
elliptic operator with parameter $\mu\in\mathbb{R}$. Moreover, the operator
\[
\boldsymbol{\sigma}\cdot i\boldsymbol{\nabla}-i\mu I_{2}:H^{1}(\mathbb{R}%
^{3},\mathbb{C}^{2})\rightarrow L^{2}(\mathbb{R}^{3},\mathbb{C}^{2})
\]
is invertible for every $\mu\in\mathbb{R}:\left\vert \mu\right\vert >0$ with
\[
\left(  \boldsymbol{\sigma}\cdot i\boldsymbol{\nabla}-i\mu I_{2}\right)
^{-1}=(\boldsymbol{\sigma}\cdot i\boldsymbol{\nabla}+i\mu I_{2})\left(
-\Delta_{2}+\mu^{2}\right)  ^{-1}I_{2}%
\]
and
\begin{equation}
\left\Vert \left(  \boldsymbol{\sigma}\cdot i\boldsymbol{\nabla}-i\mu
I_{2}\right)  ^{-1}\right\Vert _{\mathcal{B(}L^{2}(\mathbb{R}^{2}%
,\mathbb{C}^{2}),H^{1}(\mathbb{R}^{2},\mathbb{C}^{2}))}\leq\frac{C}{\left\vert
\mu\right\vert },\left\vert \mu\right\vert >0. \label{2.30'}%
\end{equation}

For the local estimates at the points $x_{0}\in\Sigma$ we use the local system
of orthogonal coordinates $y=(y_{1},y_{2})$ where the axis $y_{1}$ is directed
along the tangent to $\Sigma$ at the point $x_{0}\in\Sigma$ and the axis
$y_{3}=z$ is directed along the normal vector $\boldsymbol{\nu}_{x_{0}}$\ \ to
$\Sigma$ at the point $x_{0}\in\Sigma,$ and we take the main part of
$\mathfrak{D}_{\boldsymbol{A},\Phi}(\mu).$ Then we obtain the operator
\begin{align}
&  \mathbb{D}_{a_{+}(x_{0}),a_{-}(x_{0})}^{0}(\mu)\boldsymbol{\psi
}(y)\label{2.31'}\\
=  &  \left(  \mathbb{D}_{a_{+}(x_{0}),a_{-}(x_{0})}^{0}-i\mu\right)
\boldsymbol{\psi}(y)\nonumber\\
&  =\left\{
\begin{array}
[c]{c}%
\left(  \boldsymbol{\alpha}^{\prime}\cdot\boldsymbol{D}_{y^{\prime}}%
+i\alpha_{3}\frac{\partial}{\partial y_{3}}-i\mu I_{4}\right)
\boldsymbol{\psi(}y)=0,y\in\mathbb{R}_{+}^{3}\cup\mathbb{R}_{+}^{3}\\
a_{+}(x_{0})\boldsymbol{\psi}_{+}(y^{\prime},0)+a_{-}(x_{0})\boldsymbol{\psi
}_{-}(y^{\prime},0)=0,y^{\prime}\in\mathbb{R}^{2}%
\end{array}
\right.  ,\nonumber
\end{align}
where
\[
\boldsymbol{\alpha}^{\prime}\cdot\boldsymbol{D}_{y^{\prime}}=\alpha
_{1}D_{y_{1}}+\alpha_{2}D_{y_{2},}a_{\pm}(x_{0})=\frac{1}{2}\Gamma(x_{0})\mp
i\alpha_{3}.
\]

We investigate the invertibility of $\mathbb{D}_{a_{+}(x_{0}),a_{-}(x_{0}%
)}^{0}(\mu)$ acting from $H^{1}(\mathbb{R}_{+}^{3},\mathbb{C}^{4})\oplus
H^{1}(\mathbb{R}_{-}^{3},\mathbb{C}^{4})$ into $L^{2}(\mathbb{R}%
^{3},\mathbb{C}^{4}).$ Applying the Fourier transform with respect to
$\boldsymbol{y}^{\prime}=(y_{1},y_{2})\in\mathbb{R}^{2}$ we obtain the family
of $1-D$ parameter-dependent transmission problems on $\mathbb{R}%
\diagdown\left\{  0\right\}  $
\begin{align}
&  \mathbb{\hat{D}}_{a_{+}(x_{0}),a_{-}(x_{0})}^{0}(\xi^{\prime}%
,\mu)\label{2.31}\\
&  =\left\{
\begin{array}
[c]{c}%
\mathfrak{D}_{0}(\xi^{\prime},\mu)\psi(\xi^{\prime},\mu,z)=\left(
\boldsymbol{\alpha}^{\prime}\cdot\xi^{\prime}+i\alpha_{3}\frac{d}{dz}-i\mu
I_{4}\right)  \psi(\xi^{\prime},\mu,z),z\in\mathbb{R}\diagdown\left\{
0\right\}  ,\\
a_{+}(x_{0})\psi_{+}(\xi^{\prime},\mu,0)+a_{-}(x_{0})\psi_{-}(\xi^{\prime}%
,\mu,0)=0,\boldsymbol{\alpha}^{\prime}\cdot\xi^{\prime}=\alpha_{1}\xi
_{1}+\alpha_{2}\xi_{2}.
\end{array}
\right. \nonumber\\
\text{where }z  &  =y_{3}.\nonumber
\end{align}

Note that the operator $\mathbb{\hat{D}}_{a_{+}(x_{0}),a_{-}(x_{0})}^{0}%
(\xi^{\prime},\mu):H^{1}(\mathbb{R}\diagdown\left\{  0\right\}  ,\mathbb{C}%
^{4})\rightarrow L^{2}(\mathbb{R},\mathbb{C}^{4})$ is the Fredholm operator of
the index $0$ if $\left\vert \xi^{\prime}\right\vert ^{2}+\mu^{2}>0$. Hence
$\mathbb{D}_{a_{+}(x_{0}),a_{-}(x_{0})}^{0}(\xi^{\prime},\mu)$ is invertible
if and only if $\ker\mathfrak{B}_{a_{+}(x_{0}),a_{-}(x_{0})}(\xi^{\prime}%
,\mu)=0$ for all\ $(\xi^{\prime},\mu):$ $\left\vert \xi^{\prime}\right\vert
^{2}+\mu^{2}>0.$ We consider solutions of the equation
\[
\mathbb{\hat{D}}_{a_{+}(x_{0}),a_{-}(x_{0})}^{0}(\xi^{\prime},\mu
)\boldsymbol{\psi=0}%
\]
in the space $H^{1}(\mathbb{R}\diagdown\left\{  0\right\}  ,\mathbb{C}%
^{4})\subset L^{2}(\mathbb{R},\mathbb{C}^{4}).$ Since
\begin{equation}
\left(  \boldsymbol{\alpha}^{\prime}\cdot\xi^{\prime}+i\alpha_{3}\frac{d}%
{dz}+i\mu I_{4}\right)  \left(  \boldsymbol{\alpha}^{\prime}\cdot\xi^{\prime
}+i\alpha_{3}\frac{d}{dz}-i\mu I_{4}\right)  =\left(  \left\vert \xi^{\prime
}\right\vert ^{2}+\mu^{2}-\frac{d^{2}}{dz^{2}}\right)  I_{4} \label{2.32}%
\end{equation}
the equation%
\begin{equation}
\left(  \boldsymbol{\alpha}^{\prime}\cdot\xi^{\prime}+i\alpha_{3}\frac{d}%
{dz}-i\mu I_{4}\right)  \boldsymbol{\psi}(\xi^{\prime},\mu,z)=0 \label{2.33}%
\end{equation}
has the exponential solutions%
\begin{equation}
\boldsymbol{\psi}_{\pm}(\xi^{\prime},\mu,z)=\boldsymbol{h}_{\pm}\left(
\xi^{\prime},\mu\right)  e^{\pm\rho z},\rho=\sqrt{\left\vert \xi^{\prime
}\right\vert ^{2}+\mu^{2}}. \label{2.34}%
\end{equation}
where the vectors $\boldsymbol{h}_{\pm}\left(  \xi^{\prime},\mu\right)  $
satisfy the equation
\begin{equation}
\left(  \boldsymbol{\alpha}^{\prime}\cdot\xi^{\prime}\pm i\rho\alpha_{3}-i\mu
I_{4}\right)  \boldsymbol{h}_{\pm}\left(  \xi^{\prime},\mu\right)  =0.
\label{2.34''}%
\end{equation}

Taking into account (\ref{2.32}) we obtain that the vectors $\boldsymbol{h}%
_{\pm}\left(  \xi^{\prime},\mu\right)  \in\mathbb{C}^{4}$ have the form
\begin{equation}
\boldsymbol{h}_{\pm}\left(  \xi^{\prime},\mu\right)  =\Theta_{\pm}(\xi
^{\prime},\mu)\boldsymbol{f}_{\pm}=\left(  \boldsymbol{\alpha}^{\prime}%
\cdot\xi^{\prime}\pm i\boldsymbol{\rho}\alpha_{3}+i\mu I_{4}\right)
\boldsymbol{f}_{\pm} \label{2.34'}%
\end{equation}
where $\boldsymbol{f\pm}\in\mathbb{C}^{4}$. Let
\begin{equation}
\Lambda_{\pm}(\xi^{\prime},\mu)=\sigma^{\prime}\cdot\xi^{\prime}\pm
i\rho\sigma_{3}=\left(
\begin{array}
[c]{cc}%
\pm i\rho & \bar{\varsigma}\\
\varsigma & \mp i\rho
\end{array}
\right)  ,\varsigma=\xi_{1}+i\xi_{2} \label{2.35}%
\end{equation}
and
\[
\boldsymbol{e}_{1}=\left(
\begin{array}
[c]{c}%
1\\
0
\end{array}
\right)  ,\boldsymbol{e}_{2}=\left(
\begin{array}
[c]{c}%
0\\
1
\end{array}
\right)  ,\boldsymbol{0}=\left(
\begin{array}
[c]{c}%
0\\
0
\end{array}
\right)  .
\]
Then the vectors
\begin{align*}
\boldsymbol{h}_{1,\pm}(\xi^{\prime},\mu)  &  =\Theta_{\pm}(\xi^{\prime}%
,\mu)\left(
\begin{array}
[c]{c}%
\boldsymbol{e}_{1}\\
\boldsymbol{0}%
\end{array}
\right)  =\left(
\begin{array}
[c]{cc}%
i\mu I_{2} & \Lambda_{\pm}(\xi^{\prime},\mu)\\
\Lambda_{\pm}(\xi^{\prime},\mu) & i\mu I_{2}%
\end{array}
\right)  \left(
\begin{array}
[c]{c}%
\boldsymbol{e}_{1}\\
\boldsymbol{0}%
\end{array}
\right) \\
&  =\left(
\begin{array}
[c]{c}%
i\mu\boldsymbol{e}_{1}\\
\Lambda_{\pm}(\xi^{\prime},\mu)\boldsymbol{e}_{1}%
\end{array}
\right)  ,
\end{align*}%
\begin{align*}
\boldsymbol{h}_{2,\pm}(\xi^{\prime},\mu)  &  =\Theta_{\pm}(\xi^{\prime}%
,\mu)\left(
\begin{array}
[c]{c}%
0\\
\boldsymbol{e}_{2}%
\end{array}
\right)  =\left(
\begin{array}
[c]{cc}%
i\mu I_{2} & \Lambda_{\pm}(\xi^{\prime},\mu)\\
\Lambda_{\pm}(\xi^{\prime},\mu) & i\mu I_{2}%
\end{array}
\right)  \left(
\begin{array}
[c]{c}%
0\\
\boldsymbol{e}_{2}%
\end{array}
\right) \\
&  =\left(
\begin{array}
[c]{c}%
\Lambda_{\pm}(\xi^{\prime},\mu)\boldsymbol{e}_{2}\\
i\mu\boldsymbol{e}_{2}%
\end{array}
\right)
\end{align*}
are solutions of equation (\ref{2.34''}). Applying formulas
\begin{equation}
\Lambda_{\pm}^{\ast}(\xi^{\prime},\mu)=\Lambda_{\mp}(\xi^{\prime},\mu
),\Lambda_{\pm}^{2}(\xi^{\prime},\mu)=(-\rho^{2}+\left\vert \xi^{\prime
}\right\vert ^{2})I_{2}=-\mu^{2}I_{2} \label{2.36}%
\end{equation}
we obtain that the system of the vectors
\[
\left\{  \boldsymbol{h}_{1,+}(\xi^{\prime},\mu),\boldsymbol{h}_{2,+}%
(\xi^{\prime},\mu),\boldsymbol{h}_{1,-}(\xi^{\prime},\mu),\boldsymbol{h}%
_{2,-\pm}(\xi^{\prime},\mu)\right\}
\]
is orthogonal in $\mathbb{C}^{4}$, and $\left\{  \boldsymbol{h}_{1,\pm}%
(\xi^{\prime},\mu)e^{\pm\rho z},\boldsymbol{h}_{2,\pm}(\xi^{\prime},\mu
)e^{\pm\rho z}\right\}  $ is the fundamental system of solutions of equation
(\ref{2.33}).

The exponentially decreasing solutions of equation (\ref{2.33}) are of the
form%
\begin{equation}
\boldsymbol{\psi}=\left\{
\begin{array}
[c]{c}%
\left(  C_{1}^{+}\boldsymbol{h}_{+,1}+C_{2}^{+}\boldsymbol{h}_{+,2}\right)
e^{\rho z},z<0\\
\left(  C_{1}^{-}\boldsymbol{h}_{-,1}+C_{2}^{-}\boldsymbol{h}_{-,2}\right)
e^{-\rho z},z>0
\end{array}
\right.  . \label{2.39'}%
\end{equation}
Substituting $\boldsymbol{\psi}$ in the transmission conditions we obtain the
system of linear equations
\begin{equation}
C_{1}^{+}a_{-}(x_{0})\boldsymbol{h}_{+,1}+C_{2}^{+}a_{-}(x_{0})\boldsymbol{h}%
_{+,2}+C_{1}^{-}a_{+}(x_{0})\boldsymbol{h}_{-,1}+C_{2}^{-}a_{+}(x_{0}%
)\boldsymbol{h}_{-,2}=0 \label{2.40'}%
\end{equation}
with respect to the unknown vector $(C_{1}^{+},C_{2}^{+},C_{1}^{-},C_{2}%
^{-})\in\mathbb{C}^{4}.$ System (\ref{2.40'}) has the trivial solution if and
only if
\begin{equation}
\det\mathcal{L(}x_{0},\xi^{\prime},\mu)\neq0\text{ if }\rho^{2}=\mu
^{2}+\left\vert \xi^{\prime}\right\vert ^{2}=1 \label{e2.40}%
\end{equation}
where $\mathcal{L(}s,\xi^{\prime},\mu)$ is the matrix with columns
\[
\left\{  a_{-}(x_{0})\boldsymbol{h}_{+,1}(\xi^{\prime},\mu),a_{-}%
(x_{0})\boldsymbol{h}_{+,2}(\xi^{\prime},\mu),a_{+}(x_{0})\boldsymbol{h}%
_{-,1}(\xi^{\prime},\mu),a_{+}(x_{0})\boldsymbol{h}_{-,1}(\xi^{\prime}%
,\mu)\right\}  .
\]

\begin{itemize}
\item Condition (\ref{e2.40}) is called the local Lopatinsky-Shapiro condition
for the operator $\mathfrak{D}_{\boldsymbol{A},\Phi,a_{+},a_{-}}(\mu)$ of
parameter-dependent transmission problem (\ref{2.30}).

\item We say that the operator $\mathfrak{D}_{\boldsymbol{A},\Phi,a_{+},a_{-}%
}(\mu)$ satisfies the uniform Lopatinsky-Shapiro condition for
parameter-dependent transmission problem (\ref{2.30}) if
\begin{equation}
\inf_{x_{0}\in\Sigma,\mu^{2}+\left\vert \xi^{\prime}\right\vert ^{2}%
=1}\left\vert \det\mathcal{L(}x_{0},\xi^{\prime},\mu)\right\vert >0.
\label{2.38}%
\end{equation}

\end{itemize}

It should be noted that if condition (\ref{2.38}) is satisfied then the
operators
\[
\mathbb{D}_{a_{+}(x_{0}),a_{-}(x_{0})}^{0}(\mu):H^{1}(\mathbb{R}^{3}%
\diagdown\Sigma,\mathbb{C}^{4})\rightarrow L^{2}(\mathbb{R}^{3},\mathbb{C}%
^{4}),x_{0}\in\Sigma
\]
are invertible for every $\mu\neq0$, and
\begin{equation}
\sup_{s\in S}\left\Vert \left(  \mathbb{D}_{a_{+}(s),a_{-}(s)}^{0}%
(\mu)\right)  ^{-1}\right\Vert _{\mathcal{B}(L^{2}(\mathbb{R}^{3}%
,\mathbb{C}^{4}),H^{1}(\mathbb{R}^{3}\diagdown\Sigma,\mathbb{C}^{4}))}\leq
C\left\vert \mu\right\vert ^{-1}. \label{2.39}%
\end{equation}

\begin{theorem}
\label{t2.2} Let $\Sigma\subset\mathbb{R}^{3}$ be a $C^{2}-$uniformly regular
surface, the magnetic potential $\boldsymbol{A=}(A_{1},A_{2},A_{3})\in
L^{\infty}(\mathbb{R}^{3},\mathbb{C}^{3}),$ the electrostatic potential
$\Phi\in L^{\infty}(\mathbb{R}),$ $\Gamma\in C_{b}^{1}(\Sigma,\mathcal{B(}%
\mathbb{C}^{4})),$ and the uniformly Lopatinsky-Shapiro condition (\ref{2.38})
for parameter-dependent operator $\mathbb{D}_{\boldsymbol{A},\Phi,a_{+},a_{-}%
}(\mu),\mu\in\mathbb{R}$ be satisfied. Then there exists $\mu_{0}>0$ such that
the operator $\mathbb{D}_{\boldsymbol{A},\Phi,a_{+},a_{-}}(\mu):H^{1}%
(\mathbb{R}^{3}\diagdown\Sigma,\mathbb{C}^{4})\rightarrow L^{2}(\mathbb{R}%
^{3},\mathbb{C}^{4})$ is invertible for every $\mu\in\mathbb{R}:\left\vert
\mu\right\vert >\mu_{0}.$
\end{theorem}

\begin{proof}
The proof is similar to the proof of invertibility of elliptic
parameter-dependent boundary value problems (see for instance
\cite{AgranVishik},\cite{Agran1}, Sec.3). However, we consider the
parameter-dependent transmission problems for unbounded domains, and therefore
we need an infinite partition of unity and estimates associated with this
fact. Note that the Dirac operator $\mathfrak{D}_{\boldsymbol{A},\Phi}(i\mu)$
is a uniformly elliptic parameter-depending operator on $\mathbb{R}^{3}.$
Moreover, the Lopatinsky-Shapiro condition (\ref{2.39}) are satisfied
uniformly for every point $x\in\Sigma.$ It yields that there exists $r>0$ and
$\mu_{0}>0$ such that there exists a countable covering $\cup_{j\in\mathbb{N}%
}B_{r}(x_{j})$ of the finite multiplicity $N\geq1$ such that for every $x_{j}$
there exist operators
\[
L_{x_{j}}(\mu),R_{x_{j}}(\mu)\in\mathcal{B}(L^{2}(\mathbb{R}^{3}%
,\mathbb{C}^{4}),H^{1}(\mathbb{R}^{3}\diagdown\Sigma,\mathbb{C}^{4}))\text{ }%
\]
such that
\begin{equation}
\sup_{j\in\mathbb{N},\left\vert \mu\right\vert \geq\mu_{0}}\left\Vert
L_{x_{j}}(\mu)\right\Vert =d_{L}<\infty,\sup_{j\in\mathbb{N},\left\vert
\mu\right\vert \geq\mu_{0}}\left\Vert R_{x_{j}}(\mu)\right\Vert =d_{R}<\infty,
\label{2.40}%
\end{equation}
and
\begin{align}
L_{x_{j}}(\mu)\mathbb{D}_{\boldsymbol{A},\Phi,a_{+},a_{-}}(\mu)\eta_{j}I  &
=\eta_{j}I,\label{2.41}\\
\eta_{j}\mathbb{D}_{\boldsymbol{A},\Phi,a_{+},a_{-}}(\mu)R_{x_{j}}(\mu)  &
=\eta_{j}I\nonumber
\end{align}
for every $\eta_{j}\in C_{0}^{\infty}(B_{r}(x_{j})).$ Let
\[
\sum_{j\in\mathbb{N}}\varphi_{j}(x)=1
\]
be a partition of the unity subordinated to the covering $\cup_{j\in
\mathbb{N}}B_{r}(x_{j}).$ We set
\begin{align*}
L(\mu)u  &  =\sum_{j\in\mathbb{N}}\varphi_{j}L_{x_{j}}(\mu)\theta_{j}u,u\in
C_{0}^{\infty}(\mathbb{R}^{3},\mathbb{C}^{4}),\\
R(\mu)u  &  =\sum_{j\in\mathbb{N}}\theta_{j}R_{x_{j}}(\mu)\varphi_{j}u,
\end{align*}
where \emph{\ }$\theta_{j}\in C_{0}^{\infty}(B_{r}(x_{j})),$ $0\leq\theta
_{j}\leq1,j\in\mathbb{N},\theta_{j}\varphi_{j}=\varphi_{j}.$ One can prove
that the operators $L(\mu),R(\mu)$ are continued to the bounded operators
acting from $L^{2}(\mathbb{R}^{3},\mathbb{C}^{4})$ to $H^{1}(\mathbb{R}%
^{3}\diagdown\Sigma,\mathbb{C}^{4}))$ and
\begin{equation}
L(\mu)\mathbb{D}_{\boldsymbol{A},\Phi,a_{+},a_{-}}(\mu)=I+T_{1}(\mu
),\mathbb{D}_{\boldsymbol{A},\Phi,a_{+},a_{-}}(\mu)R(\mu)=I+T_{2}(\mu)
\label{2.42'}%
\end{equation}
where
\begin{align}
\left\Vert T_{1}(\mu)\right\Vert  &  \leq C(1+\left\vert \mu\right\vert
)^{-1},\label{2.44'}\\
\left\Vert T_{2}(\mu)\right\Vert  &  \leq C(1+\left\vert \mu\right\vert
)^{-1},\nonumber
\end{align}
and $C>0$ is independent of $\mu.$ Hence the operator $\mathbb{D}%
_{\boldsymbol{A},\Phi,a_{+},a_{-}}(\mu)$ is invertible for $\left\vert
\mu\right\vert $ large enough.
\end{proof}

\begin{example}
\label{e2} Let \
\[
Q_{\sin}=\Gamma(s)\delta_{\Sigma}=\left(
\begin{array}
[c]{cc}%
\gamma(s) & 0\\
0 & \kappa(s)
\end{array}
\right)  \delta_{\Sigma}%
\]
$\gamma,\kappa\in C_{b}^{1}(\Sigma).$ We set
\[
\left\{  \mathfrak{f}_{1}=a_{-}(s)\boldsymbol{h}_{+,1},\mathfrak{f}%
_{2}(s)=a_{-}(s)\boldsymbol{h}_{+,2},\mathfrak{f}_{3}(s)=a_{+}%
(s)\boldsymbol{h}_{-,1},\mathfrak{f}_{4}(s)=a_{+}(s)\boldsymbol{h}%
_{-,1}\right\}  ,
\]
where
\begin{align}
\mathfrak{f}_{1}  &  =a_{-}\boldsymbol{h}_{1,+}=\left(
\begin{array}
[c]{cc}%
\gamma I_{2} & -i\sigma_{3}\\
-i\sigma_{3} & \kappa I_{2}%
\end{array}
\right)  \left(
\begin{array}
[c]{c}%
i\mu\boldsymbol{e}_{1}\\
\Lambda_{+}\boldsymbol{e}_{1}%
\end{array}
\right)  =\left(
\begin{array}
[c]{c}%
i\eta\gamma e_{1}+i\Lambda_{-}\boldsymbol{e}_{1}\\
\mu e_{1}+\kappa\Lambda_{+}\boldsymbol{e}_{1}%
\end{array}
\right)  ,\label{f2.42}\\
\mathfrak{f}_{2}  &  =a_{-}\boldsymbol{h}_{2,+}=\left(
\begin{array}
[c]{cc}%
\gamma I_{2} & -i\sigma_{3}\\
-i\sigma_{3} & \kappa I_{2}%
\end{array}
\right)  \left(
\begin{array}
[c]{c}%
\Lambda\boldsymbol{e}_{2}\\
i\mu\boldsymbol{e}_{2}%
\end{array}
\right)  =\left(
\begin{array}
[c]{c}%
\gamma\Lambda_{+}\boldsymbol{e}_{2}-\mu\boldsymbol{e}_{2}\\
-i\Lambda_{-}\boldsymbol{e}_{2}+i\mu\kappa\boldsymbol{e}_{2}%
\end{array}
\right)  ,\nonumber\\
\mathfrak{f}_{3}  &  =a_{+}\boldsymbol{h}_{1,-}=\left(
\begin{array}
[c]{cc}%
\gamma I_{2} & -i\sigma_{3}\\
-i\sigma_{3} & \kappa I_{2}%
\end{array}
\right)  \left(
\begin{array}
[c]{c}%
i\mu\boldsymbol{e}_{1}\\
\Lambda_{-}\boldsymbol{e}_{1}%
\end{array}
\right)  =\left(
\begin{array}
[c]{c}%
i\eta\mu e_{1}-i\Lambda_{+}\boldsymbol{e}_{1}\\
-\mu e_{1}+\kappa\Lambda_{-}\boldsymbol{e}_{1}%
\end{array}
\right)  ,\nonumber\\
\mathfrak{f}_{4}  &  =a_{+}\boldsymbol{h}_{2,-}=\left(
\begin{array}
[c]{cc}%
\gamma I_{2} & -i\sigma_{3}\\
-i\sigma_{3} & \kappa I_{2}%
\end{array}
\right)  \left(
\begin{array}
[c]{c}%
\Lambda_{-}\boldsymbol{e}_{2}\\
i\mu\boldsymbol{e}_{2}%
\end{array}
\right)  =\left(
\begin{array}
[c]{c}%
\eta\Lambda_{-}\boldsymbol{e}_{2}+\mu\boldsymbol{e}_{2}\\
i\Lambda_{+}\boldsymbol{e}_{2}+i\mu\kappa\boldsymbol{e}_{2}%
\end{array}
\right)  .\nonumber
\end{align}
It follows from (\ref{f2.42}) that
\begin{align}
\mathfrak{f}_{1}\cdot\mathfrak{f}_{2}  &  =\mathfrak{f}_{1}\cdot
\mathfrak{f}_{3}=\mathfrak{f}_{2}\cdot\mathfrak{f}_{4}=\mathfrak{f}_{3}%
\cdot\mathfrak{f}_{4}=0,\\
\mathfrak{f}_{1}^{2}  &  =\mathfrak{f}_{2}^{2}=2(1+\gamma^{2})\left\vert
\xi^{\prime}\right\vert ^{2},\mathfrak{f}_{3}^{2}=\mathfrak{f}_{4}%
^{2}=2(1+\epsilon^{2})\left\vert \xi^{\prime}\right\vert ^{2},\nonumber\\
\mathfrak{f}_{1}\cdot\mathfrak{f}_{4}  &  =\mathfrak{f}_{2}\cdot
\mathfrak{f}_{3}=2i(\gamma+\epsilon)^{2}\left\vert \xi^{\prime}\right\vert
^{2}.\nonumber
\end{align}
Applying formulas (\ref{2.36}) we obtain that the Gramm determinant $G$ of the
system $\left\{  \mathfrak{f}_{j}\right\}  _{j=1}^{4}$ is%
\begin{equation}
G=\det\left(  \mathfrak{f}_{i}\cdot\mathfrak{f}_{j}\right)  _{i,j=1}%
^{4}=16\rho^{8}(1-\gamma\epsilon)^{4},\rho^{2}=\mu^{2}+\left\vert \xi^{\prime
}\right\vert ^{2}. \label{2.43'}%
\end{equation}
Hence the Lopatinsky-Shapiro condition holds uniformly on $\Sigma$ if \
\begin{equation}
\inf_{s\in\Sigma,}\left\vert 1-\gamma(s)\kappa(s)\right\vert >0.
\label{2.43''}%
\end{equation}
IHence if the condition (\ref{2.43''}) holds the operator
\[
\mathfrak{D}_{\boldsymbol{A},\Phi,a_{+},a_{-}}(\mu):H^{1}(\mathbb{R}%
^{3}\diagdown\Sigma,\mathbb{C}^{4})\rightarrow L^{2}(\mathbb{R}^{3}%
\diagdown\Sigma,\mathbb{C}^{4}),\mu\in\mathbb{R}%
\]
is invertible for large values of $\left\vert \mu\right\vert .$
\end{example}

\subsection{ Self-adjointness of the operator $\mathcal{D}_{\boldsymbol{A}%
,\Phi,a_{+},a_{b-}}$}

Now we consider the self-adjointness of the unbounded operator $\mathcal{D}%
_{\boldsymbol{A},\Phi,a_{+},a_{-}}$ in $L^{2}(\mathbb{R}^{3},\mathbb{C}^{4})$
defined by the Dirac operator
\[
\mathfrak{D}_{\boldsymbol{A,}\Phi}=\boldsymbol{\alpha}\cdot\left(
\boldsymbol{D}_{x}+\boldsymbol{A}(x)\right)  +\alpha_{0}m+\Phi(x)I_{4}%
\]
with the domain
\[
H_{a_{+},a_{-}}^{1}(\mathbb{R}^{3}\diagdown\Sigma,\mathbb{C}^{4})=\left\{
\boldsymbol{u}\in H^{1}(\mathbb{R}^{3}\diagdown\Sigma,\mathbb{C}^{4}%
):a_{+}(s)\boldsymbol{u}_{+}(s)+a_{-}(s)\boldsymbol{u}_{-}(s)=\boldsymbol{0,}%
s\in\Sigma\right\}
\]
where
\[
a_{+}(s)=\frac{1}{2}\Gamma(s)-i\boldsymbol{\nu}\cdot\boldsymbol{\alpha}%
,a_{-}(s)=\frac{1}{2}\Gamma(s)+i\boldsymbol{\nu}\cdot\boldsymbol{\alpha}.
\]
We set
\begin{align*}
\left\langle \boldsymbol{u,v}\right\rangle _{L^{2}(\mathbb{R}^{3}%
,\mathbb{C}^{4})}  &  =\int_{\mathbb{R}^{3}}\boldsymbol{u}(x)\cdot
\boldsymbol{v}(x)dx,\\
\left\langle \boldsymbol{u,v}\right\rangle _{L^{2}(\Sigma,\mathbb{C}^{4})}  &
=\int_{\Sigma}\boldsymbol{u}(s)\cdot\boldsymbol{v}(s)ds
\end{align*}
are the scalar products in $L^{2}(\mathbb{R}^{3},\mathbb{C}^{4}),L^{2}%
(\Sigma,\mathbb{C}^{4}),$ respectively.

\begin{theorem}
\label{t2.3} Let $\Sigma\subset\mathbb{R}^{3}$ be the $C^{2}$-uniformly
regular surface, the vector potential $\boldsymbol{A}\in L^{\infty}%
(\mathbb{R}^{3},\mathbb{R}^{4})$ and the scalar potential $\Phi\in L^{\infty
}(\mathbb{R}^{3})$ be real-valued, and $\Gamma=\left(  \Gamma_{ij}\right)
_{i,j=1}^{4}$ be the Hermitian matrix with elements $\Gamma_{ij}\in C_{b}%
^{1}(\Sigma).$ We assume that the uniform Lopatinsky-Shapiro conditions
(\ref{2.38}) for the parameter-dependent problem holds. \ Then the operator
$\mathcal{D}_{\boldsymbol{A},\Phi,a_{+},a_{-}}$ with domain $H_{a_{+},a_{-}%
}^{1}(\mathbb{R}^{3}\diagdown\Sigma,\mathbb{C}^{4})$ is self-adjoint in
$L^{2}(\mathbb{R}^{3},\mathbb{C}^{4}).$
\end{theorem}

\begin{proof}
At first, we prove that the operator $\mathcal{D}_{\boldsymbol{A},\Phi
,a_{+},a_{-}}$ is symmetric. Let $\boldsymbol{u,v}\in H_{a_{+},a_{-}}%
^{1}(\mathbb{R}^{3}\diagdown\Sigma,\mathbb{C}^{4}).$ Integrating by parts we
obtain
\begin{align*}
&  \left\langle \mathbb{D}_{\boldsymbol{A,}\Phi}\boldsymbol{u},\boldsymbol{v}%
\right\rangle _{L^{2}(\mathbb{R}^{3},\mathbb{C}^{4})}-\left\langle
\boldsymbol{u},\mathbb{D}_{\boldsymbol{A,}\Phi}\boldsymbol{v}\right\rangle
_{L^{2}(\mathbb{R}^{3},\mathbb{C}^{4})}\\
&  =\left\langle \left(  -i\boldsymbol{\alpha}\cdot\boldsymbol{\nu}\right)
\boldsymbol{u}_{+},\boldsymbol{v}_{+}\right\rangle _{L^{2}(\Sigma
,\mathbb{C}^{4})}-\left\langle \left(  -i\boldsymbol{\alpha}\cdot
\boldsymbol{\nu}\right)  \boldsymbol{u}_{-},\boldsymbol{v}_{-}\right\rangle
_{L^{2}(\Sigma,\mathbb{C}^{4})}\\
&  =\frac{1}{2}\left\langle -i\boldsymbol{\alpha}\cdot\boldsymbol{\nu}\left(
\boldsymbol{u}_{+}-\boldsymbol{u}_{-}\right)  ,\boldsymbol{v}_{+}%
-\boldsymbol{v}_{-}\right\rangle _{L^{2}(\Sigma,\mathbb{C}^{4})}-\frac{1}%
{2}\left\langle \boldsymbol{u}_{+}-\boldsymbol{u}_{-},-i\boldsymbol{\alpha
\cdot\nu(v}_{+}-\boldsymbol{v}_{-}\right\rangle _{L^{2}(\Sigma,\mathbb{C}%
^{4})}.
\end{align*}
Thaking into account equality (\ref{2.3}) we obtain that
\begin{align}
&  \left\langle \mathbb{D}_{\boldsymbol{A,}\Phi}\boldsymbol{u},\boldsymbol{v}%
\right\rangle _{L^{2}(\mathbb{R}^{3},\mathbb{C}^{4})}-\left\langle
\boldsymbol{u,}\mathbb{D}_{\boldsymbol{A,}\Phi}\boldsymbol{v}\right\rangle
_{L^{2}(\mathbb{R}^{3},\mathbb{C}^{4})}\label{2.46}\\
&  =-\frac{1}{4}\left\langle \Gamma\left(  \boldsymbol{u}_{+}+\boldsymbol{u}%
_{-}\right)  ,\boldsymbol{v}_{+}-\boldsymbol{v}_{-}\right\rangle
_{L^{2}(\Sigma,\mathbb{C}^{4})}+\frac{1}{4}\left\langle \boldsymbol{u}%
_{+}+\boldsymbol{u}_{-},\Gamma(\boldsymbol{v}_{+}-\boldsymbol{v}%
_{-}\right\rangle _{L^{2}(\Sigma,\mathbb{C}^{4})}.\nonumber
\end{align}
Since the matrix $\Gamma$ is Hermitian for every $s\in\Sigma$ the right part
side in (\ref{2.46}) is $0.$ Hence
\[
\left\langle \mathbb{D}_{\boldsymbol{A,}\Phi}\boldsymbol{u,v}\right\rangle
_{L^{2}(\mathbb{R}^{3},\mathbb{C}^{4})}=\left\langle \boldsymbol{u}%
,\mathbb{D}_{\boldsymbol{A,}\Phi}\boldsymbol{v}\right\rangle _{L^{2}%
(\mathbb{R}^{3},\mathbb{C}^{4})}%
\]
for every $\boldsymbol{u,v}\in H_{a_{+},a_{-}}^{1}(\mathbb{R}^{3}%
\diagdown\Sigma,\mathbb{C}^{4})$. \ The uniform Lopatinsky-Shapiro condition
(\ref{2.38}) yields a priori estimate (\ref{2.20'}) which implies that the
operator $\mathcal{D}_{\boldsymbol{A},\Phi,a_{+},a_{-}}$is closed. Moreover,
Theorem \ref{t2.2} yields that the deficiency indices of $\mathcal{D}%
_{\boldsymbol{A},\Phi,a_{+},a_{-}}$ are equal $0.$ Hence (see for instance
\cite{BirSol}, page 100) the operator $\mathcal{D}_{\boldsymbol{A},\Phi
,a_{+},a_{-}}$is self-adjoint.
\end{proof}

\subsubsection{Self-adjointness of Dirac operator with electrostatic and
Lorentz $\delta-$shell interactions}

As application of Theorem \ref{t2.3} we consider the Dirac operator with
singular potential
\[
Q_{\sin}=(\eta(s)+\tau(s))\delta_{\Sigma}=\left(
\begin{array}
[c]{cc}%
\left(  \eta(s)+\tau(s)\right)  I_{2} & 0\\
0 & \left(  \eta(s)-\tau(s)\right)  I_{2}%
\end{array}
\right)
\]
of the electrostatic and Lorentz $\delta-$shell interactions.

\begin{theorem}
\label{t2.4} Let \ $\Sigma$ be a uniformly regular $C^{2}-$surface. We assume
that the vector potential $\boldsymbol{A}\in L^{\infty}(\mathbb{R}%
^{3},\mathbb{R}^{4})$, the scalar potential $\Phi\in L^{\infty}(\mathbb{R}%
^{3})$ and real-valued, the function $\eta(s),\tau(s)(\in C_{b}^{1}(\Sigma))$
are real valued. Moreover, we assume that
\begin{equation}
\inf_{s\in S}\left\vert \eta^{2}(s)-\tau^{2}(s)-4\right\vert >0. \label{2.47}%
\end{equation}
Then the unbounded operator $\mathcal{D}_{\boldsymbol{A},\Phi,a_{+},a_{-}}$
defined by the Dirac operator $\mathfrak{D}_{\boldsymbol{A},\Phi}$ with domain
$H_{a_{+},a_{-}}^{1}(\mathbb{R}^{3}\diagdown\Sigma,\mathbb{C}^{4})$ where
$a_{\pm}(s)=\eta(s)I_{4}+\tau(s)\alpha_{0}\mp i\boldsymbol{\alpha\cdot\nu
}(s),s\in\Sigma$ is self-adjoint in $L^{2}(\mathbb{R}^{3},\mathbb{C}^{4}).$
\end{theorem}

\begin{proof}
We set in Example \ref{e2} $\gamma=\eta+\tau,\kappa=\eta-\tau.$ Applying
condition (\ref{2.43''}) we obtain that (\ref{2.47}) is the uniform
Lopatinsky-Shapiro condition $\mathcal{D}_{\boldsymbol{A},\Phi,a_{+},a_{-}}$
It should be noted that condition (\ref{2.47}) ensures the fulfillment of the
uniform Lopatinsky-Shapiro condition for the operator $\mathbb{D}%
_{\boldsymbol{A},\Phi,a_{+},a_{-}}(\mu)$ of parameter-dependent transmission
problem associated with potential $Q_{\sin}(s)=(\eta(s)+\tau(s))\delta
_{\Sigma}.$ Moreover, the unbounded operator $\mathcal{D}_{\boldsymbol{A}%
,\Phi,a_{+},a_{-}}$ with domain $H_{a_{+},a_{-}}^{1}(\mathbb{R}^{3}%
\diagdown\Sigma,\mathbb{C}^{4})$ is symmetric in $L^{2}(\mathbb{R}%
^{3},\mathbb{C}^{4})$. Hence Theorem \ref{t2.4} yields that $\mathcal{D}%
_{\boldsymbol{A},\Phi,a_{+},a_{-}}$ is a self-adjoint operator.
\end{proof}

\section{Fredholm theory of transmission problems associated with Dirac
operator with singular potentials}

We consider the Fredholm property of the transmission operator%
\begin{equation}
\mathbb{D}_{\boldsymbol{A},\Phi,a_{+},a_{-}}\boldsymbol{u}(x)=\left\{
\begin{array}
[c]{c}%
\mathfrak{D}_{\boldsymbol{A},\Phi}\boldsymbol{u}(x),x\in\mathbb{R}%
^{3}\diagdown\Sigma,\\
a_{+}\mathcal{(}s)\boldsymbol{u}_{+}(s)+a_{-}\mathcal{(}s)\boldsymbol{u}%
_{-}(s)=\boldsymbol{0},s\in\Sigma;
\end{array}
\right.  \label{e4.1}%
\end{equation}
associated with the Dirac operator with singular potential $\mathfrak{D}%
_{\boldsymbol{A},\Phi,Q_{\sin}}=\mathfrak{D}_{\boldsymbol{A},\Phi}+Q_{\sin} $
where
\begin{equation}
\mathfrak{D}_{\boldsymbol{A},\Phi}u(x)=\left(  \boldsymbol{\alpha}%
\cdot(\boldsymbol{D}_{x}+\boldsymbol{A}(x))+\alpha_{0}m+\Phi(x)\right)
\boldsymbol{u}(x),x\in\mathbb{R}^{3}\diagdown\Sigma\label{e4.2}%
\end{equation}
and $Q_{\sin}=\Gamma\delta_{S}.$ We assume that $\Sigma$ is a connected
$C^{2}$-surface being the common boundary of the domains $\Omega_{\pm}.$
Moreover, $\Sigma$ is a closed compact surface or unbounded uniformly regular surface.

\subsection{Simonenko's local principle}

\begin{itemize}
\item Let $\chi\in C_{0}^{\infty}(\mathbb{R}^{n})$ be such that $0\leq
\chi(x)\leq1,\chi(x)=1$ for $\left\vert x\right\vert \leq1,$ and $\chi(x)=0$
for $\left\vert x\right\vert \geq2,$ and $\chi_{R}(x)=\chi(\frac{x}{R})$,
$\psi_{R}(x)=1-\chi_{R}(x).$ The function $\psi_{R}$ is called the cut-off
function of infinity. \ Let $B_{\varepsilon}(x_{0})=\left\{  x\in
\mathbb{R}^{n}:\left\vert x-x_{0}\right\vert <\varepsilon\right\}  .$ We say
that $\varphi_{x_{0}}$ is a cut-off function of $B_{\varepsilon}(x_{0})$ if
$\varphi_{x_{0}}\in C_{0}^{\infty}(B_{\varepsilon}(x_{0})),0\leq\varphi
_{x_{0}}(x)\leq1$ and $\varphi_{x_{0}}(x)=1$ if $x\in B_{\varepsilon/2}%
(x_{0}).$

\item We denote by $\widetilde{\mathbb{R}^{3}}$ the compactification of
$\mathbb{R}^{3}$ obtained by the adjoint to every ray $l_{\omega}=\left\{
x\in\mathbb{R}^{3}:x=t\omega,t>0,\omega\in S^{2}\right\}  $ the infinitely
distant point $\vartheta_{\omega}$. The topology in $\widetilde{\mathbb{R}%
^{3}}$ is introduced such that $\widetilde{\mathbb{R}^{3}}$ becomes
homeomorphic to the unit closed ball $\bar{B}_{1}(0).$ $\ $The fundamental
system of neighborhoods of the point $\vartheta_{\omega_{0}}$ is formed by the
conical sets $U_{\omega_{0},R}=\digamma_{\omega_{0}}\times(R,+\infty)$ where
$R>0$ and $\digamma_{\omega_{0}}$ is a neighborhood of the point $\omega_{0}$
on the unit sphere $S^{2}.$ We define the cut-off function $\varphi_{x_{0}}$
of the infinitely distant point $x_{0}=\vartheta_{\omega_{9}}$ as
$\varphi_{x_{0}}=\varphi_{\omega_{0}}(\frac{x}{\left\vert x\right\vert }%
)\psi_{R}(x)$ where $\varphi_{\omega_{0}}(\omega)\in C_{0}^{\infty}%
(\digamma_{\omega_{0}})$ and $\varphi_{\omega_{0}}(\omega)=1$ in a
neighborhood $\digamma_{\omega_{0}}^{\prime}$ such that $\overline
{\digamma_{\omega_{0}}^{\prime}}\subset\digamma_{\omega_{0}}.$
\end{itemize}

\begin{definition}
\label{ed4.1} We say that the operator
\[
\mathbb{D}_{\boldsymbol{A},\Phi,a_{+},a_{-}}:H^{1}(\mathbb{R}^{3}\diagdown
S,\mathbb{C}^{4})\rightarrow L^{2}(\mathbb{R}^{3},\mathbb{C}^{4})
\]
\ is locally invertible at the point $x_{0}\in\widetilde{\mathbb{R}^{3}}$ if
there exists $\ $a neighborhood $U_{x_{0}}$ of the point $x_{0}$ and the
operators
\[
\mathcal{L}_{x_{0}},\mathcal{R}_{x_{0}}\in\mathcal{B(}L^{2}(\mathbb{R}%
^{3},\mathbb{C}^{4}),H^{1}(\mathbb{R}^{3}\diagdown\Sigma,\mathbb{C}^{4}))
\]
such that
\begin{equation}
\mathcal{L}_{x_{0}}\mathbb{D}_{\boldsymbol{A},\Phi,a_{+},a_{-}}\varphi_{x_{0}%
}I=\varphi_{x_{0}}I,\varphi_{x_{0}}\mathbb{D}_{\boldsymbol{A},\Phi,a_{+}%
,a_{-}}\mathcal{R}_{x_{0}}=\varphi_{x_{0}}I, \label{e4.3}%
\end{equation}
where $\varphi_{x_{0}}$ is the cut-off function of the point $x_{0}.$
\end{definition}

\begin{proposition}
\label{ep4.1}(\cite{Sim1967}, \cite{Rab1972}) The operator
\[
\mathbb{D}_{\boldsymbol{A},\Phi,a_{+},a_{-}}:H^{1}(\mathbb{R}^{3}%
\diagdown\Sigma,\mathbb{C}^{4})\rightarrow L^{2}(\mathbb{R}^{3},\mathbb{C}%
^{4})
\]
is Fredholm if and only if $\mathbb{D}_{\boldsymbol{A},\Phi,a_{+},a_{-}}$ is
locally invertible at every point $x\in\widetilde{\mathbb{R}^{3}}.$
\end{proposition}

\subsection{Transmission on compact closed surfaces}

We consider the Fredholm property of the operator $\mathbb{D}_{\boldsymbol{A}%
,\Phi,a_{+},a_{-}}$ if $\Sigma$ is a compact closed $C^{2}-$surface.

We assume that the vector-valued potential $\boldsymbol{A}=(A_{1},A_{2}%
,A_{3})$ and the electric potentials $\Phi$ are such that
\begin{equation}
A_{j},\Phi\in C_{b}^{1}(\mathbb{R}^{3}), \label{4.0}%
\end{equation}%
\begin{equation}
\Gamma=\left(  \Gamma_{ij}\right)  _{i,j=1}^{4},\Gamma_{ij}\in C_{b}%
^{1}(\Sigma). \label{4.1}%
\end{equation}
Let $\ f\in C_{b}^{1}(\mathbb{R}^{3}),\mathbb{R}^{3}\ni g_{m}\rightarrow
\infty.$ We consider the functional sequence $\left\{  f(x+g_{m})\right\}  .$
This sequence is uniformly bounded and eqicontinuous on every compact set
$K\subset\mathbb{R}^{3}.$ By the Arzel\`{a}--Ascoli \ Theorem there exists a
subsequence $h_{m}$ of $g_{m}$ and the limit function $f^{h}\in C_{b}%
(\mathbb{R}^{3})$ such that
\begin{equation}
\lim_{h_{m}\rightarrow\infty}\sup_{x\in K}\left\vert f(x+h_{m})-f^{h}%
(x)\right\vert =0 \label{4.3}%
\end{equation}
for every compact set $K\subset\mathbb{R}^{3}.$

Since $A_{j},\Phi\in C_{b}^{1}(\mathbb{R}^{3}),$ for every sequence
$\mathbb{R}^{3}\ni g_{m}\rightarrow$ $\vartheta_{\omega}$ there exists a
subsequence $h_{m}$ and limit functions $A_{j}^{h},j=1,2,3,\Phi^{h}$
belonging~to $C_{b}(\mathbb{R}^{n})$ defined by formula (\ref{4.3}). The
operator $\mathfrak{D}_{\boldsymbol{A},\Phi}^{h}=\mathfrak{D}_{\boldsymbol{A}%
^{h},\Phi^{h}}$ is called the \textit{limit operator} of $\mathfrak{D}%
_{\boldsymbol{A},\Phi}.$ We denote by $Lim_{\vartheta_{\omega}}\mathfrak{D}%
_{\boldsymbol{A},\Phi}$ the set of all limit operators of $\mathfrak{D}%
_{\boldsymbol{A},\Phi}$ defined by sequences $h_{m}\rightarrow\vartheta
_{\omega},$ and \ we set
\begin{equation}
Lim\mathfrak{D}_{\boldsymbol{A},\Phi}=%
{\displaystyle\bigcup\limits_{\vartheta_{\omega}\in\mathbb{R}_{\infty}^{3}}}
Lim_{\vartheta_{\omega}}\mathfrak{D}_{\boldsymbol{A},\Phi}. \label{4.1'}%
\end{equation}

\begin{theorem}
\label{t4.1} Let condition (\ref{4.0}) hold, $\Sigma$ be a $C^{2}-$compact
closed surface dividing $\Omega_{+}$ and $\Omega_{-},$ and Lopatinsky-Shapiro
condition (\ref{2.19}) hold at every point $s\in\Sigma.$ Then $\mathbb{D}%
_{\boldsymbol{A},\Phi,a_{+},a_{-}}:H^{1}(\mathbb{R}^{3}\diagup\Sigma
,\mathbb{C}^{4})\rightarrow L^{2}\left(  \mathbb{R}^{3},\mathbb{C}^{4}\right)
$ is the Fredholm operator if and only if all limit operators $\mathfrak{D}%
_{\boldsymbol{A},\Phi}^{h}\in Lim\mathfrak{D}_{\boldsymbol{A},\Phi}$ are
invertible from $H^{1}(\mathbb{R},\mathbb{C}^{4})$ into $L^{2}\left(
\mathbb{R}^{3},\mathbb{C}^{4}\right)  .$
\end{theorem}

\begin{proof}
Note that $\mathbb{D}_{\boldsymbol{A},\Phi,a_{+},a_{-}}$ locally coinsides
with $\mathfrak{D}_{\boldsymbol{A},\Phi}$ outside $\Sigma.$ Because
\ $\mathfrak{D}_{\boldsymbol{A},\Phi}$ is the elliptic operator on
$\mathbb{R}^{3}$, $\mathfrak{D}_{\boldsymbol{A},\Phi}$ is a locally invertible
operator at every point $x\in\mathbb{R}^{3}\diagdown\Sigma$ . According to the
Lopatinsky-Shapiro condition (\ref{2.19}) $\mathbb{D}_{\boldsymbol{A}%
,\Phi,a_{+},a_{-}}$ is a locally invertible operator at every point
$x\in\Sigma.$ Proposition \ref{ep4.1} yields that $\mathbb{D}_{\boldsymbol{A}%
,\Phi,a_{+},a_{-}}$ is a Fredholm operator if and only if $\mathbb{D}%
_{\boldsymbol{A},\Phi,a_{+},a_{-}}$ is locally invertible at every infinitely
distant point $\vartheta_{\omega}\in\mathbb{R}_{\infty}^{3}.$ The operator
$\mathbb{D}_{\boldsymbol{A},\Phi,a_{+},a_{-}}$ coincides with the operator
$\mathfrak{D}_{\boldsymbol{A},\Phi}$ near every point $\vartheta_{\omega}.$
Applying the results of book \cite{RRS} and paper \cite{Ra1} we obtain that
$\mathfrak{D}_{\boldsymbol{A},\Phi}$ is locally invertible at infinity if and
only if all limit operators $\mathfrak{D}_{\boldsymbol{A}^{h},\Phi^{h}}\in
Lim\mathfrak{D}_{\boldsymbol{A},\Phi}$ are invertible.
\end{proof}

\begin{corollary}
\label{co4.1} Let conditions of Theorem \ref{t4.1} hold. Then%
\begin{equation}
sp_{ess}\mathcal{D}_{\boldsymbol{A},\Phi,a_{+},a_{-}}=%
{\displaystyle\bigcup\limits_{\mathfrak{D}_{\boldsymbol{A}^{h},\Phi^{h}}\in
Lim\mathfrak{D}_{\boldsymbol{A},\Phi}}}
sp\mathcal{D}_{\boldsymbol{A}^{h},\Phi^{h}} \label{f4.1}%
\end{equation}
where $\mathcal{D}_{\boldsymbol{A}^{h},\Phi^{h}}$ is unbounded in
$L^{2}(\mathbb{R}^{3},\mathbb{C}^{4})$ generated by $\mathfrak{D}%
_{\boldsymbol{A}^{h},\Phi^{h}}.$
\end{corollary}

\subsection{Slowly oscillating at infinity potentials}

We say that a function $f\in C_{b}^{1}(\mathbb{R}^{3})$ is slowly oscillating
at infinity if
\[
\lim_{x\rightarrow\infty}\partial_{x_{j}}f(x)=0,j=1,2,3
\]
We denote the class of slowly oscillating at infinity functions by
$SO^{1}(\mathbb{R}^{3}).$ Note that if $f\in SO_{\infty}(\mathbb{R}^{3})$ and
$\mathbb{R}^{3}\ni h_{m}\rightarrow\infty$ is such that there exists a limit
function $f^{h}$ in the sense of formula (\ref{4.3}). Then the function
$f^{h}$ is a constant.

We consider the operator $\mathbb{D}_{\boldsymbol{A},\Phi,a_{+},a_{-}}$ for
$A_{j},\Phi\in SO^{1}(\mathbb{R}^{3}),$ and we assume that $A_{j},\Phi$ are
real-valued functions, $\Sigma$ is the compact closed $C^{2}-$surface, and the
interaction matrix $\Gamma=(\Gamma_{ij})_{i,j=1}^{4}$ is such that
$\Gamma_{ij}\in C^{1}(\Sigma)$ and real-valued functions.

Then the limit operators are of the form:
\begin{equation}
\mathbb{D}_{\boldsymbol{A},\Phi,a_{+},a_{-}}^{h}=\mathfrak{D}_{A^{h},\Phi^{h}%
}=\alpha\cdot(\boldsymbol{D+A}^{h})\boldsymbol{+}\alpha_{0}m+\Phi^{h}I_{4}
\label{f4.3}%
\end{equation}
where $\boldsymbol{A}^{h}\in\mathbb{R}^{3},\Phi^{h}\in\mathbb{R}.$ Note that
\begin{equation}
sp\mathfrak{D}_{A^{h},\Phi^{h}}=(-\infty,\Phi^{h}-\left\vert m\right\vert ]%
{\displaystyle\bigcup}
[\Phi^{h}+\left\vert m\right\vert ,+\infty) \label{f4.4}%
\end{equation}

Then according formula (\ref{f4.1}) we obtain that
\begin{equation}
sp_{ess}\mathcal{D}_{\boldsymbol{A},\Phi,,a_{+},a_{-}}=(-\infty,M_{\Phi}%
^{\sup}-\left\vert m\right\vert ]%
{\displaystyle\bigcup}
[M_{\Phi}^{\inf}+\left\vert m\right\vert ,+\infty) \label{4.3'}%
\end{equation}
where
\[
M_{\Phi}^{\sup}=\limsup_{x\rightarrow\infty}\Phi(x),M_{\Phi}^{\inf}%
=\liminf_{x\rightarrow\infty}\Phi(x).
\]

Formula (\ref{4.3'}) yields that $sp_{ess}\mathcal{D}_{\boldsymbol{A}%
,\Phi,,a_{+},a_{-}}$ is independent from the slowly oscillating at infinity
magnetic potential $\boldsymbol{A.}$ \ Moreover,\ \ if $\left\vert
m\right\vert >0$ and $M_{\Phi}^{\sup}-M_{\Phi}^{\inf}<2\left\vert m\right\vert
,$ then $sp_{ess}\mathcal{D}_{\boldsymbol{A},\Phi,,a_{+},a_{-}}$ has a gap
$(M_{\Phi}^{\sup}-\left\vert m\right\vert ,M_{\Phi}^{\inf}+\left\vert
m\right\vert )$ which could contain the discrete spectrum of $\mathcal{D}%
_{\boldsymbol{A},\Phi,,a_{+},a_{-}}.$ In the opposite case $M_{\Phi}^{\sup
}-M_{\Phi}^{\inf}\geq2\left\vert m\right\vert $
\[
sp_{ess}\mathcal{D}_{\boldsymbol{A},\Phi,a_{+},a_{-}}=(-\infty,+\infty).
\]
Formula (\ref{4.3'}) yields that $sp_{ess}\mathcal{D}_{\boldsymbol{A}%
,\Phi,,a_{+},a_{-}}$ is independent from the slowly oscillating at infinity
magnetic potential $\boldsymbol{A.}$ \ Moreover,\ \ if $\left\vert
m\right\vert \geq0$ and $M_{\Phi}^{\sup}-M_{\Phi}^{\inf}<2\left\vert
m\right\vert ,$ then $sp_{ess}\mathcal{D}_{\boldsymbol{A},\Phi,,a_{+},a_{-}}$
has a gap $(M_{\Phi}^{\sup}-\left\vert m\right\vert ,M_{\Phi}^{\inf
}+\left\vert m\right\vert )$ in the essential spectrum which could contain the
discrete spectrum of $\mathcal{D}_{\boldsymbol{A},\Phi,,a_{+},a_{-}}.$ In the
opposite case: $M_{\Phi}^{\sup}-M_{\Phi}^{\inf}\geq2\left\vert m\right\vert $
\[
sp_{ess}\mathcal{D}_{\boldsymbol{A},\Phi,a_{+},a_{-}}=(-\infty,+\infty).
\]

\subsection{Transmission on unbounded surfaces with conic structure at
infinity}

Let $\Omega_{+}\subset\mathbb{R}^{3}$ be an connected open domain with a
$C^{2}-$boundary $\Sigma$, $\Omega_{-}=\mathbb{R}^{3}\diagdown\Omega
_{+},\Sigma$ \ has the conic structure at infinity, that is there exists $R>0$
such that if $x_{0}\in\Sigma$ and $x_{0}>R$ the ray $\left\{  x\in
\mathbb{R}^{3}:x=tx_{0},t>0\right\}  \subset\Sigma$. Let $\tilde{\Omega}_{\pm
},\tilde{\Sigma}$ be the compactifications of the sets $\Omega_{\pm},\Sigma$
in the topology of $\widetilde{\mathbb{R}^{3}}.$

We consider the operator $\mathbb{D}_{\boldsymbol{A},\Phi,a_{+},a_{-}}$ for
$A_{j},\Phi\in SO^{1}(\mathbb{R}^{3})$ and for the interaction matrix
$\Gamma=\left(  \Gamma_{ij}\right)  _{i,j=1}^{4}$ with $\Gamma_{ij}\in
C(\tilde{\Sigma})\cap C^{1}(\Sigma).$

We define the limit operators of the operator $\mathcal{D}_{\boldsymbol{A}%
,\Phi,a_{+},a_{-}}$ as follows:

\begin{itemize}
\item If $\vartheta_{\omega}\notin\Sigma_{\infty}=\tilde{\Sigma}%
\diagdown\Sigma,$ then the limit operators defined by the sequence
$h_{m}\rightarrow\vartheta_{\omega}$ are of the form (\ref{f4.3}) with the
spectrum given by formula (\ref{f4.4}).

\item Let $\vartheta_{\omega}\in\Sigma_{\infty}=\tilde{\Sigma}\diagdown
\Sigma,$ $\ $ $l_{\omega}^{R}=\left\{  x\in\mathbb{R}^{3}:x=t\omega
,t>R\right\}  $ and $T_{\vartheta_{\omega}}$ be the tangent plane to $\Sigma$
at the$\ $ ray $l_{\omega}^{R}$ and $\boldsymbol{\nu}(\omega)$ is the outgoing
normal vector to $\Omega_{+}$ at the points of the ray $l_{\omega}^{R}.$ We
denote by $\mathbb{R}_{\pm,\vartheta_{\omega}}^{3}$ the half-spaces in
$\mathbb{R}^{3}$ with common boundary $T_{\vartheta_{\omega}}.$ Then following
to the paper \cite{Ra2} the limit operators defined by the sequences
$h_{m}\rightarrow\vartheta_{\omega}$ are of the form
\begin{equation}
\mathbb{D}_{\boldsymbol{A},\Phi,a_{+}(\vartheta_{\omega}),a_{-}(\vartheta
_{\omega})}^{h}u(x)=\left\{
\begin{array}
[c]{c}%
\mathfrak{D}_{\boldsymbol{A}^{h},\Phi^{h}}\boldsymbol{u}(x),x\in\mathbb{R}%
^{3}\diagdown T_{\vartheta_{\omega}},\\
a_{+}\mathcal{(}\vartheta_{\omega})\boldsymbol{u}_{+}(s)+a_{-}\mathcal{(}%
\vartheta_{\omega})\boldsymbol{u}_{-}(s)=0,s\in T_{\vartheta_{\omega}}%
\end{array}
\right.  , \label{5.1}%
\end{equation}
where
\begin{align*}
a_{+}\mathcal{(}\vartheta_{\omega})  &  =\frac{1}{2}\Gamma(\vartheta_{\omega
})-i\alpha\cdot\boldsymbol{\nu}(\omega),a_{-}\mathcal{(}\vartheta_{\omega
})=\frac{1}{2}\Gamma(\vartheta_{\omega})+i\alpha\cdot\boldsymbol{\nu}%
(\omega),\\
\Gamma(\vartheta_{\omega})  &  =\lim_{\Sigma\ni s\rightarrow\vartheta_{\omega
}}\Gamma(s).
\end{align*}

\end{itemize}

\begin{theorem}
\label{te5.1} Let $\Sigma$ be a $C^{2}-$surface with conical structure at
infinity, $A_{j},\Phi\in SO^{1}(\mathbb{R}^{3})$ be real-valued functions and
the interaction matrix $\Gamma=\left(  \Gamma_{i,j}\right)  _{i,j=1}^{4}$ be
Hermitian with $\Gamma_{i,j}\in C\left(  \tilde{\Sigma}\right)  \cap C_{b}%
^{1}(\Sigma),$ and the uniform Lopatinsky-Shapiro condition be satisfied.
Then: (i) the operator $\mathcal{D}_{\boldsymbol{A},\Phi,a_{+},a_{-}}$ is
self-adjoint, (ii) $\mathcal{D}_{\boldsymbol{A},\Phi,a_{+},a_{-}}$ is a
Fredholm operator if and only if: (a) for every $\vartheta_{\omega}%
\notin\Sigma_{\infty}$ all limit operators $\mathfrak{D}_{\boldsymbol{A},\Phi
}^{h}$ defined by the sequences $h_{m}\rightarrow\vartheta_{\omega}$ are
invertible from $H^{1}(\mathbb{R}^{3},\mathbb{C}^{4})\rightarrow
L^{2}(\mathbb{R}^{3},\mathbb{C}^{4});$ (b) for every $\vartheta_{\omega}%
\in\Sigma_{\infty}$ all limit operators $\mathbb{D}_{\boldsymbol{A}^{h}%
,\Phi^{h},a_{+}\mathcal{(}\vartheta_{\omega}),a_{-}\mathcal{(}\vartheta
_{\omega})}$ $\in Lim_{\vartheta_{\omega}}\mathfrak{D}_{\boldsymbol{A}%
,\Phi,a_{+},a_{-}}$ are invertible from $H^{1}(\mathbb{R}^{3}\diagdown
\Sigma,\mathbb{C}^{4})\rightarrow L^{2}(\mathbb{R}^{3},\mathbb{C}^{4}).$
\end{theorem}

\begin{proof}
The self-adjointness of operator $\mathcal{D}_{\boldsymbol{A},\Phi,a_{+}%
,a_{-}}$ follows from Theorem \ref{te2.3}. The operator $\mathbb{D}%
_{\boldsymbol{A},\Phi,a_{+},a_{-}}$ is locally invertible at every point
$x\in\mathbb{R}^{3}$ because the Dirac operator $\mathfrak{D}_{A,\Phi}$ is
elliptic and the Lopatinsky-Shapiro condition holds at every point $x\in
\Sigma.$ Hence according Proposition \ref{ep4.1} $\mathbb{D}_{\boldsymbol{A}%
,\Phi,a_{+},a_{-}}$ is a Fredholm operator if and only if $\mathbb{D}%
_{\boldsymbol{A},\Phi,a_{+},a_{-}}$ is locally invertible at every infinitely
distant point $\vartheta_{\omega}\in\mathbb{R}_{\infty}^{3}$. Following to the
monograph \cite{RRS}, the papers \cite{Ra1}, \cite{Ra2} we obtain that
conditions (a) and (b) are necessary and sufficient for the local
invertibility of $\mathbb{D}_{\boldsymbol{A},\Phi,a_{+},a_{-}}$ at every
infinitely distant point.
\end{proof}

\begin{corollary}
\label{co5.1} Let the conditions of Theorem \ref{te5.1} hold. Then
\begin{equation}
sp_{ess}\mathcal{D}_{\boldsymbol{A},\Phi,a_{+},a_{-}}=%
{\displaystyle\bigcup\limits_{\mathcal{D}_{\boldsymbol{A},\Phi,a_{+},a_{-}%
}^{h}\in Lim\mathcal{D}_{\boldsymbol{A},\Phi,a_{+},a_{-}}}}
sp\mathcal{D}_{\boldsymbol{A},\Phi,a_{+},a_{-}}^{h} \label{e5.1}%
\end{equation}
where $\mathcal{D}_{\boldsymbol{A},\Phi,a_{+},a_{-}}^{h}$ are unbounded
operators associated with the above defined transmission operators
$\mathbb{D}_{\boldsymbol{A},\Phi,a_{+},a_{-}}^{h}.$
\end{corollary}

Note that
\begin{equation}
sp\mathcal{D}_{\boldsymbol{A},\Phi,a_{+},a_{-}}^{h}=sp\mathfrak{D}%
_{\boldsymbol{A}^{h},\Phi^{h}}=sp\mathfrak{D}_{\boldsymbol{0},\Phi^{h}%
}=\left(  -\infty,\Phi^{h}-\left\vert m\right\vert \right]
{\displaystyle\bigcup}
\left[  \Phi^{h}+\left\vert m\right\vert ,+\infty\right)  \label{5.3}%
\end{equation}
for the sequences $h_{m}\rightarrow\vartheta_{\omega}\in\mathbb{R}_{\infty
}^{3}\diagdown\Sigma_{\infty}.$

We consider now the spectrum of operators $\mathbb{D}_{\boldsymbol{A}%
^{h}\boldsymbol{,}\Phi^{h},a_{+}(\vartheta_{\omega}),a_{-}(\vartheta_{\omega
})}$ defined by the sequences $h_{m}\rightarrow\vartheta_{\omega}\in
\Sigma_{\infty}$. Note that
\[
sp\mathbb{D}_{\boldsymbol{A}^{h}\boldsymbol{,}\Phi^{h},a_{+}(\vartheta
_{\omega}),a_{-}(\vartheta_{\omega})}=sp\mathbb{D}_{\boldsymbol{0,}\Phi
^{h},a_{+}(\vartheta_{\omega}),a_{-}(\vartheta_{\omega})}%
\]
where
\[
\mathbb{D}_{\boldsymbol{0,}\Phi^{h},a_{+}(\vartheta_{\omega}),a_{-}%
(\vartheta_{\omega})}\boldsymbol{u}(x)=\left\{
\begin{array}
[c]{c}%
\left(  \alpha\cdot\boldsymbol{D}+\alpha_{0}m+\Phi^{h}\right)  \boldsymbol{u}%
(x),x\in\mathbb{R}^{3}\diagdown\mathbb{T}_{\vartheta_{\omega}},\\
a_{+}(\vartheta_{\omega})\boldsymbol{u}_{+}(s)+a_{-}(\vartheta_{\omega
})\boldsymbol{u}_{-}(s)=0,s\in\mathbb{T}_{\vartheta_{\omega}}.
\end{array}
\right.
\]
Without loss of generality we assume that $\mathbb{T}_{\vartheta_{\omega}%
}=\mathbb{R}_{x^{\prime}}^{2}=\left\{  x=(x^{\prime},x_{3}):x_{3}=0\right\}
.$ Then after the Fourier transform with respect to $x^{\prime}\in
\mathbb{R}^{2}$ we obtain the family of one-dimensional Dirac operators
depended on the parameter $\xi^{\prime}\in\mathbb{R}^{2}$
\begin{align}
&  \mathbb{\hat{D}}_{\boldsymbol{0,}\Phi^{h},a_{+}(\vartheta_{\omega}%
),a_{-}(\vartheta_{\omega})}\left(  \xi^{\prime}\right)  \boldsymbol{\hat{u}%
}(\xi^{\prime},z)\label{5.4}\\
&  =\left\{
\begin{array}
[c]{c}%
\left(  \alpha^{\prime}\cdot\xi^{\prime}+i\alpha_{3}\frac{d}{dz}+\alpha
_{0}m+\Phi^{h}\right)  \boldsymbol{\hat{u}}(\xi^{\prime},z),z\in
\mathbb{R}\diagdown\left\{  0\right\}  ,\xi^{\prime}\in\mathbb{R}^{2}\\
a_{+}(\vartheta_{\omega})\boldsymbol{\hat{u}}_{+}(\xi^{\prime},0)+a_{-}%
(\vartheta_{\omega})\boldsymbol{u}_{-}(\xi^{\prime},0)=0,\xi^{\prime}%
\in\mathbb{R}^{2}.
\end{array}
\right.  .\nonumber
\end{align}
The $1-D$ Dirac operator
\begin{equation}
\mathbb{\hat{D}}_{\boldsymbol{0,}\Phi^{h},a_{+}(\vartheta_{\omega}%
),a_{-}(\vartheta_{\omega})}\left(  \xi^{\prime}\right)  \boldsymbol{\varphi
}(z)=\left\{
\begin{array}
[c]{c}%
\left(  \alpha^{\prime}\cdot\xi^{\prime}+i\alpha_{3}\frac{d}{dz}+\alpha
_{0}m+\Phi^{h}\right)  \boldsymbol{\varphi}(z),z\in\mathbb{R\diagdown}\left\{
0\right\} \\
a_{+}(\vartheta_{\omega})\boldsymbol{\varphi}_{+}(0)+a_{-}(\vartheta_{\omega
})\boldsymbol{\varphi}_{-}(0)=0
\end{array}
\right.  \label{5.5}%
\end{equation}
has the essential spectrum
\[
sp_{ess}\mathbb{\hat{D}}_{\boldsymbol{0,}\Phi^{h},a_{+}(\vartheta_{\omega
}),a_{-}(\vartheta_{\omega})}\left(  \xi^{\prime}\right)  =\left(
-\infty,\Phi^{h}-\sqrt{\left\vert \xi^{\prime}\right\vert ^{2}+\left\vert
m\right\vert ^{2}}\right]
{\displaystyle\bigcup}
\left[  \Phi^{h}+\sqrt{\left\vert \xi^{\prime}\right\vert ^{2}+\left\vert
m\right\vert ^{2}},+\infty\right)
\]
and a possible finite discrete spectrum on the interval $\left(  \Phi
^{h}-\sqrt{\left\vert \xi^{\prime}\right\vert ^{2}+m^{2}},\Phi^{h}%
+\sqrt{\left\vert \xi^{\prime}\right\vert ^{2}+m^{2}}\right)  .$ Hence
\begin{align*}
sp\mathbb{D}_{\boldsymbol{0},\Phi^{h},a_{+}(\vartheta_{\omega}),a_{-}%
(\vartheta_{\omega})}==  &  \left(  -\infty,-\left\vert m\right\vert +\Phi
^{h}\right]
{\displaystyle\bigcup}
\left[  \left\vert m\right\vert +\Phi^{h},+\infty\right) \\
&
{\displaystyle\bigcup\limits_{\xi^{\prime}\in\mathbb{R}^{2}}}
sp_{dis}\mathbb{\hat{D}}_{\boldsymbol{0},\Phi^{h},a_{+}^{\vartheta_{\omega}%
},a_{-}^{\vartheta_{\omega}}}(\xi^{\prime})%
{\displaystyle\bigcap}
(-\left\vert m\right\vert +\Phi^{h},\left\vert m\right\vert +\Phi^{h}).
\end{align*}
Applying Corollary \ref{co5.1} we obtain the following result.\ \ \ \ \ \

\begin{theorem}
\label{te5.2} Let the conditions of Theorem \ref{te5.1} hold. Then
\begin{align*}
&  sp_{ess}\mathcal{D}_{\boldsymbol{A},\Phi,,a_{+},a_{-}}\\
&  =(-\infty,M_{\Phi}^{\sup}-\left\vert m\right\vert ]%
{\displaystyle\bigcup}
[M_{\Phi}^{\inf}+\left\vert m\right\vert ,+\infty)\\
&
{\displaystyle\bigcup\limits_{\xi^{\prime}\in\mathbb{R}^{2}}}
{\displaystyle\bigcup\limits_{\vartheta_{\omega}\in\Sigma_{\infty}}}
{\displaystyle\bigcup\limits_{h_{m}\rightarrow\vartheta_{\omega}}}
sp_{dis}\mathbb{\hat{D}}_{\boldsymbol{0},\Phi^{h},a_{+}^{\vartheta_{\omega}%
},a_{-}^{\vartheta_{\omega}}}(\xi^{\prime})%
{\displaystyle\bigcap}
(-\left\vert m\right\vert +\Phi^{h},\left\vert m\right\vert +\Phi^{h}).
\end{align*}

\end{theorem}

\begin{remark}
The calculation of the essential spectrum of $\mathcal{D}_{\boldsymbol{A,}%
\Phi,a_{+},a_{-}}$ is simplified if $\ $%
\[
\lim_{s\rightarrow\infty}\Gamma(s)=0.
\]
In this case, the limit operators defined by the sequences $h_{m}%
\rightarrow\vartheta_{\omega}$ are of the form
\begin{equation}
\mathbb{D}_{\boldsymbol{0,}\Phi^{h},a_{+}^{\vartheta_{\omega}},a_{-}%
^{\vartheta_{\omega}}}\boldsymbol{u}(x)=\left\{
\begin{array}
[c]{c}%
\left(  \boldsymbol{\alpha}\cdot\boldsymbol{D}_{x}+\alpha_{0}m+\Phi
^{h}\right)  \boldsymbol{u}(x),x\in\mathbb{R}^{3}\diagdown\mathbb{T}%
_{\vartheta_{\omega}}\\
\boldsymbol{u}_{+}(s)=\boldsymbol{u}_{-}(s),s\in\mathbb{T}_{\vartheta_{\omega
}}.
\end{array}
\right.  . \label{e5.10}%
\end{equation}
It yields that $sp_{dis}\mathbb{D}_{\boldsymbol{0,}\Phi^{h},a_{+}%
^{\vartheta_{\omega}},a_{-}^{\vartheta_{\omega}}}=\varnothing$ and if
conditions of Theorem \ref{te5.1} impliy formula (\ref{e5.1}).
\end{remark}

\subsection{Electrostatic and Lorentz scalar $\delta-$shell interaction for
conic at infinity interaction surfaces}

Let $\Sigma$ be a conic at infinity $C^{2}-$surface. We consider the operator
$\mathcal{D}_{A,\Phi,a_{+},a_{-}},$ where $A_{j},\Phi\in SO^{1}(\mathbb{R}%
^{3})$ are real-valued functions. The singular potential
\begin{equation}
Q_{\sin}=\Gamma\delta_{\Sigma}=\left(  \eta I_{4}+\tau\alpha_{4}\right)
\delta_{\Sigma} \label{5.11}%
\end{equation}
where $\eta,\tau\in C_{b}^{1}(\Sigma)$ and $\lim_{\Sigma\ni s\rightarrow
\infty}f(s)=0$. Moreover, we assume that $\eta,\tau$ are real-valued functions
and the Lopatinsky-Shapiro condition
\begin{equation}
\eta^{2}(s)-\tau^{2}(s)\text{ }\neq4 \label{5.13}%
\end{equation}
be satisfied for every point $s\in\Sigma.$ Then the operator $\mathcal{D}%
_{A,\Phi,a_{+},a_{-}}$ is self-adjoint in $L^{2}(\mathbb{R}^{3},\mathbb{C}%
^{4})$ and $sp_{ess}\mathcal{D}_{A,\Phi,a_{+},a_{-}}$ is given by formula
(\ref{e5.1}).

\bigskip

\begin{thebibliography}{99}                                                                                               %


\bibitem {Agran1}Agranovich,M.S.: Elliptic boundary problems, in Partial
Differential Equations,IX, Agranovich,M.S.,Egorov,Y.V.,Shubin,M.A. (Eds.)
Springer, Berlin-Heidelberg-New York, (2010).

\bibitem {AgranVishik}Agranovich,M.S., Vishik,M.I.: Elliptic problems with a
parameter and parabolic problems of general forms. Uspekhi Mat. Nauk.
1964,219; 63--161; English trans. Russian Math. Surveys. 1964; 19; 53--157.

\bibitem {ALB}Albeverio,S.,Gesztesy,F.,Hoegh-Krohn,F.,Holden,H.: Solvable
Models in Quantum Mechanics, with an Appendix by Pavel Exner, 2nd edition, AMS
Chelsea Publishing, Providence, RI (2005).

\bibitem {AlbKurasov}Albeverio,S.,Kurasov,P.: Singular Perturbations of
Differential Operators and Schr\"{o}dinger Type Operators, Cambridge
Univ.Press, (2000).

\bibitem {A-V}Arrizabalaga, N.,Mas,A.and Vega,L.: Shell interactions for Dirac
operators. J. Math. Pures Appl. (9), 102(4):617--639, 2014

\bibitem {BirSol}Birman,M.Sh.,Solomjak,M.Sh.: Spectral Theory of Selfadjoint
Operators in Hilbert Spaces. Reidel, Dordrecht (1987).

\bibitem {Bogolubov}Bogolubov,N.N., Shirkov,D.V.: Quantum \ Fields,
Benjamin/Cummings Publishing Company \ Inc. (1982).

\bibitem {BrunGeyler}Bruening,J.,Geyler,V.,Pankrashkin,K.: Spectra of
self-adjoint extensions and applications to solvable Schr\"{o}dinger
operators, Rev. Math. Phys. 20, 1--70 (2008).

\bibitem {BEKS94}Brasche,J.F.,ExnerN. Arrizabalaga, A. Mas, and L. Vega: Shell
interactions for Dirac operators. J. Math. Pures Appl. (9), 102(4):617--639, (2014).

\bibitem {BEL}Behrndt,J.,Exner,P.,Lotoreichik,V.: Essential spectrum of
Schr\"{o}dinger operators with $\delta$-interactions on the union of compact
Lipschitz surfaces. PAMM Proc. Appl. Math. Mech.; 13: 523 -- 524, (2013).

\bibitem {BEL1}Behrndt,J.,Langer,M.,Lotoreichik,V.: Schr\"{o}dinger operators
with $\delta$ and $\delta^{\prime}$-potentials supported on surfaces. Ann.
Henri Poincar\'{e}, 14: 385--423, (2013).

\bibitem {BEL3}Behrndt,J.,Exner,P.,Lotoreichik,V.: Schr\"{o}dinger operators
with $\delta-$ and $\delta^{\prime}-$ interactions on Lipschitz surfaces and
chromatic numbers of associated partitions. Rev. Math. Phys. 26 (1450015) (
[43 pages], (2014).

\bibitem {BEHL}Behrndt,J.,Exner,P.,Holzmann,M.,Lotoreichik,V.: On the spectral
properties of Dirac operators with electrostatic $\delta$-shell interactions,
J. Math.Pures Appl.111, 47--78, (2018).

\bibitem {BEHL1}Behrndt,J.,Exner,P.,Holzmann,M.,Lotoreichik,V: On Dirac
operators in $\mathbb{R}^{3}$ with electrostatic and Lorentz scalar $\delta
$-shell interactions, Quantum Stud.: Math. Found.,
https://doi.org/10.1007/s40509-019-00186-6, (2019).

\bibitem {BF}Berezin,F.A.,Faddeev,L.D.: A remark on Schr\"{o}dinger operators
with a singular potentials. Soviet Math.Dokl. ;137: 1011-1014 (1961).

\bibitem {BusStolz}Buschmann,D.,Stolz,G.,Weidmann,J.: One-dimensional
Schr\"{o}dinger operators with local point interactions, J. Reine Angew.Math.
467, 169--186 (1995).

\bibitem {Bjorken}Bjorken, J.D., Drell,S.D.: Relativistic Quantum Mechanics,
McGraw-Hill Book Company, New York St. Louis San Francisco Toronto London
Sydney (1964)

\bibitem {Guilkey}Gilkey,P. B., Invariance theory, the heat equation, and the
Atiyah-Singer index theorem, seconded., CRC Press, Boca Raton, FL, 199

\bibitem {Kur}Kurasov,P.: Distribution theory for discontinuous test functions
and differential operators with generalized coefficients. Journal of
Math.Anal.and Appl. 201, 297-323, (1996). \

\bibitem {LionsMagenes}Lions, J. L., Magenes E.: Non-Homogeneous Boundary
Value Problems and Applications, vol.1, Springer-Verlag Berlin Heidelberg New
York (1972).

\bibitem {OurmieresVega}Ourmieres-Bonafos,Th.,Vega,L.: A strategy for
self-adjointness of Dirac operators: Applications to the MIT BAG model and
shell interactions. Publ.Mat. 62, 397-437,(2018).

\bibitem {Pan}Moroianu, A.,Ourm\'{\i}eres-Bonafos,Th.,Pankrashkin, K.: Dirac
operators on surfaces large mass limits, Journal Math. Pures et Appliquees, V.
102, Is. 4, Pages 617 - 639 (2014).

\bibitem {Ourm}Ourmieros-Bonafos, Th., Pizzichlllo,F., Dirac operators and
shell interactions: a survey, arXiv:1902.03901v1 [math-ph] 11 Feb 2019.

\bibitem {RRS1}Rabinovich,V.S.,Roch,S.,Silbermann,B.: Band-dominated operators
with operator-valued coefficients, their Fredholm properties and finite
sections, Integr. Equ. Oper. Theory, 40:3, 342-381 (2001).

\bibitem {RRS}Rabinovich,V.S.,Roch,S.,Silbermann,B.: Limit Operators and their
Applications in Operator Theory, In ser.Operator Theory: Advances and
Applications, vol 150, Birkh\"{a}user Verlag, (2004).

\bibitem {Ra1}Rabinovich,V.S.: Essential spectrum of perturbed
pseudodifferential operators. Applications to the Schr\"{o}dinger,
Klein-Gordon, and Dirac operators, Russ. J. Math. Physics, 12:1, 62-80, (2005).

\bibitem {Ra2}Rabinovich,V.S.: Transmission problems for conical and
quasi-conical at infinity domains, Applicable Analysis, Vol. 94, No. 10,
2077--2094 (2015).

\bibitem {Ra2017}Rabinovich,V.S.: Essential spectrum of Schro\"{o}dinger
operators with $\delta-$interactions on unbounded surfaces, Math. Notes, 102:
5, 698--709, (2017).

\bibitem {Ra2018}Rabinovich,V.S.: Schr\"{o}dinger operators with interactions
on unbounded surfaces: Math. Meth. Appl. Sci. Math Meth Appl Sci.;42:
4981--4998 (2019).

\bibitem {Rab1972}Rabinovich,V.S.: Pseudodifferential operators on a class of
noncompact manifolds, Math. USSR-Sb., 18:1, 45-59, (1972).

\bibitem {Sim1967}Simonenko I.B.: Operators of convolution type in cones,
Math. USSR-Sb., 3:2, 279-293, (1967).

\bibitem {Thaller}Thaller, B.: The Dirac Equation, Springer-Verlag, Berlin
Heidelberg NewYork London (1956).
\end{thebibliography}
\end{document}